\documentclass[11pt]{article}

\usepackage{geometry,showlabels}
\usepackage{xr-hyper}
\usepackage{times,enumitem}
\usepackage{url}
\usepackage[colorlinks = true,linkcolor =blue,citecolor = blue, urlcolor = blue]{hyperref}

\usepackage{bigints}
\usepackage{amsfonts,amsmath,amssymb,amsthm}
\usepackage{wrapfig}
\usepackage{verbatim,float,url}
\usepackage{graphicx,psfrag}
\usepackage{subcaption}
\usepackage[round]{natbib}   % omit 'round' option if you prefer square brackets
\usepackage[english]{babel}
\usepackage{cancel}
\usepackage{color}

\usepackage{mathtools}

\newtheorem{theorem}{Theorem}
\newtheorem{lemma}[theorem]{Lemma}

\usepackage{mdwlist}	
\theoremstyle{remark}

\theoremstyle{definition}

\usepackage{tikz}  
\usetikzlibrary{trees}

\graphicspath{{figs/}}

\usepackage{mathtools}
\DeclarePairedDelimiter\ceil{\lceil}{\rceil}
\DeclarePairedDelimiter\floor{\lfloor}{\rfloor}
\newcommand{\indep}{\rotatebox[origin=c]{90}{$\models$}}

\theoremstyle{plain}
\newtheorem{assumption}{Assumption}

\begin{document}

\title{Graphon estimation via  nearest neighbor algorithm and 2D  fused lasso denoising }

%\title{$K$-NN total variation regression}
\author{
	Oscar Hernan Madrid Padilla$^{1}$\\
	{\tt omadrid@berkeley.edu}\\
	\and
%	James Sharpnack$^{2}$ \\
%	{\tt  jsharpna@ucdavis.edu} \\
%	\and
	Yanzhen Chen $^{2}$ \\
	{\tt  imyanzhen@ust.hk} \\
%	\and
%	Daniela Witten$^{3}$ $^{4}$ \\
%	{\tt  dwitten@uw.edu} \\
	\begin{tabular}{c}
		$^{1}$ Department of Statistics, UC Berkeley\\
%		$^{2}$ Department of Statistics, UC Davis\\
%		$^{3}$ Department of Statistics, University of Washington\\
%		$^{4}$ Department of Biostatistics, University of Washington\\ 
		$^{2}$  Department of  ISOM, Hong Kong University of Science and Technology
	\end{tabular}
}

\date{\today}

\maketitle

 \begin{abstract}
 	We propose  a class of  methods for graphon estimation based on exploiting connections with non-parametric regression.  The idea is to construct an ordering of the nodes in the network, similar in spirit  to \cite{chan2014consistent}.  However, rather than only considering orderings based on the empirical  degree as in \cite{chan2014consistent},  we use the nearest  neighbor algorithm  which is an approximative solution to  the traveling 
 	salesman  problem. This in turn can  handle  general distances $\hat{d}$   between  the nodes, which allows us to incorporate  rich information of the network. Once an ordering is constructed, we formulate a 2D  grid graph denoising problem that we solve through  fused lasso regularization. For particular  choices of the metric $\hat{d}$,  we show that  the corresponding  two--step  estimator can attain competitive rates  when the true model  is the stochastic block model, and when the underlying  graphon is piecewise H\"{o}lder or   it has  bounded variation.  

 \end{abstract}

\section{Introduction}
\label{sec:intro}

Suppose we are given measurements  $A_{i,j} \in \{0,1\}$  with   $i,j \in \{1,\ldots,n\}$, and
with
\begin{equation}
\label{eqn:model1}
   \begin{array}{lll}
   A_{j,i}   & =&  A_{i,j},   \,\,\,\,\,\,\,\,\,\,\,\,\,\,  \forall  i,j \in\{1,\ldots,n\},\\
   A_{i,i}  &=& 0,   \,\,\,\,\,\,\,\,\,\,\,\,\,\,    \forall  i \in\{1,\ldots,n\},\\
   A_{i,j} &\sim &\text{Bernoulli}\left(\,\theta_{i,j}^*\right),   \,\,\,\,\,\,\,\, \forall  i < j, \,\,\, i,j \in \{1,\ldots,n\},	  \\
   \theta_{i,j}^* &=& f_0(\xi_i,\xi_j) \\
   \xi_ i  &\sim^{ \text{i.i.d} }  &  U[0,1],  \,\,\,\,i \in \,\{1,\ldots,n\},
   \end{array}
\end{equation}
where $f_0 \,:\, [0,1] \times [0,1]   \,\rightarrow [0,1]$   is a function that might depend on $n$. Moreover, the indices $i $  and  $j$  denote  the nodes of a network  with $n$  vertices, and    $A_{i,j} \,\,=\,\,1$
indicates that there is an edge between  $i$ and $j$. The goal  is  to estimate $\theta^*_{i,j}$,  the probability that there is a link between $i$  and $j$,  under   structural assumptions  of the function $f_0$. 

The model  described in (\ref{eqn:model1})  has attracted significant attention in the statistics and  machine  learning communities.  This is due to the increasing  amount of data that can be represented as a  binary  graph model. For instance, emails between individuals, social networks connections, financial networks   where an edge indicates a transaction between individuals, and many more.

  The goal of this paper is to study a class of methods  that can perform well in practice, yet  enjoying attractive  statistical  properties under different classes of functions. We are particularly interested in settings where $f_0$ is   piecewise H\"{o}lder, has bounded variation, or is an instance of the  stochastic block model. Our idea is based on the connection  between nonparametric regression  and graphon estimation, see  \cite{gao2015rate}  for a discussion.  
  Roughly speaking, our approach proceeds  as follows. We start constructing a 2D grid graph $ G_{2D}$ with set of nodes $V = \{1,\ldots,n\} \times \{ 1,\ldots,n \}$. To construct the  set of  edges we first find a  permutation $\hat{\tau}$ of the set $\{1,\ldots,n\}$ (we will discus below  some choices for $\hat{\tau}$). With this permutation in hand,  the set of edges of $G_{2D}$, namely $E $, is such  that $(i,j) \in V$  is connected  to $(i^{\prime},j^{\prime})\in V$ if and only if 
  \[
       \left\vert  \hat{\tau}^{-1}(i) - \hat{\tau}^{-1}(i^{\prime}) \right\vert   \,+\,  \left\vert  \hat{\tau}^{-1}(j) - \hat{\tau}^{-1}(j^{\prime}) \right\vert \,=\, 1.\,\,\,\,\,\,\
  \]
  Using the graph $G_{2D}$, we  define the total variation of  $\theta \in  \mathbb{R}^{n \times n }$, with respect to such graph, as
  \[
    \| D \theta \|_1\,=\,  \underset{( (i,j),(i^{\prime},j^{\prime})    )  \in E }{\sum }\, \left\vert \theta_{i,j}  - \theta_{i^{\prime},j^{\prime}}     \right \vert.  
  \]
  This similar to the notion  of total variation as studied in \cite{sadhanala2016total}, but  the 2D grid graph is based on the permutation $P$  that  we will learn from the data. Our hope is to construct  a reasonable permutation $P$  in that sense that the true parameter has small total variation, or that $\|  D \theta^* \|_1  <<  n^2 $. If this were achieved, then it would make sense to estimate $\theta^*$ as the solution to 
  %and 
  %and  find a permutation $P$ of the set $\{1,\ldots,n\}$ (we will discus later the details of to find $P$). With this permutation in hand, 
  
 % we construct  a 2D grid graph  of size $n \times n$  by using information from the data $A$, and then  we solve an optimization problem given by
 \begin{equation}
 	\label{eqn:first_problem}
 	 \underset{\theta  \in \mathbb{R}^{n \times n }}{\min}\,\, \frac{1}{2}\|A\,-\,\theta \|_F^2  \,\,+\,\, \lambda\,\| D \theta \|_1, 	 
 \end{equation}
  where  $\| \cdot \|_F$  is the Frobenius norm, and $\lambda > 0 $ is a tuning parameter that is meant to emphasize the believe that the total variation along $G_{2D} $ should be small.
  %and $\|\cdot\|_1$ is the usual  $\ell_1$ norm, see the notation section below.   The detailed  construction of $D$ will be given later.
  
  Thus,   we perform  total  variation denosing on the adjacency matrix $A$, treating its entries as nodes of  carefully designed grid graphs. The   first type of grid  that we consider is constructed as follows: We think of the elements of
   $\{1,\ldots,n\}$  as cities with distances  induced by the metric  from  \cite{zhang2015estimating}. Thus the distance  between $i$  and $j$ is
   \[
      \hat{d}_I(i,j)  \,\,=\,\,  \frac{1}{n}\,  \underset{k \neq i,j }{\max} \left\vert    \sum_{l =1}^n (A_{i,l}  - A_{j,l} )A_{k,l}  \right\vert. 
   \] 
     Using such metric, we run  the  nearest neighbor (NN) algorithm, which is an approximate solution to  the traveling salesman problem, see  for instance  \cite{rosenkrantz1977analysis}.  This  gives a permutation  $\hat{\tau}(1), \ldots,\hat{\tau}(n) $ of the nodes that can be used  to   embed the matrix $A$  in a 2D grid. With the embedding   we can then solve  a problem   in the form of (\ref{eqn:first_problem}). We refer to the resulting procedure as NN-FL,  given that it combines the NN algorithm with fused lasso regularization.  Here,  fused lasso   refers to the estimator for denoising problems introduced by \cite{tibshirani2005sparsity},  whose  predecessor is the total variation denoising estimator  from 
      \cite{rudin1992nonlinear}. Equation (\ref{eqn:first_problem})  is actually  a  graph  fused lasso  denoising problem as studied from a theoretical perspective in \cite{padilla2016dfs} and \cite{hutter2016optimal}, and from an algorithmic point of view in \cite{barbero2014modular} and \cite{tansey2015fast}. We exploit these algorithms in order to efficiently compute our estimator.

     %  More  recently,  grid graph denoising has been  studied using discrete differences,  a particular case of which is  (\ref{eqn:first_problem}). This discretization  idea was introduced to the statistics community by \cite{tibshirani2005sparsity} under the name of fused lasso. Since then, efficient algorithms have been developed  for grid graph denosing such as in \cite{barbero2014modular} and \cite{tansey2015fast}. We exploit these algorithms in order to efficiently compute our estimator.

The second  approach that we study considers  an alternative to  the metric  form \cite{zhang2015estimating}. Basically,  we simply  sort  the degrees of the nodes  just as in \cite{chan2014consistent}. However,  once the ordering is obtained, we do not use the  penalty  in (\ref{eqn:sorting})  as in \cite{chan2014consistent} but we rather apply fused lasso denoising. Hence  we refer to the resulting procedure as SAS-FL,  to emphasize  that is a minor modification of the sort and smooth  (SAS)  method from \cite{chan2014consistent}.  This  small   difference between  SAS-FL and the method from \cite{chan2014consistent} allows us to study the former
 on classes of bounded variation.  We refer the reader  to \cite{sadhanala2016total}  for  a discussion on some advantages  of using   the fused lasso on grids.

   Finally,  we  will also  discuss other possible  metrics choices  such as  those based  on external data (independent of $A$),  or  the $\ell_1$  distance on the columns of $A$.
  % Finally,   we emphasize that (\ref{eqn:first_problem}) is a particular instance of the class of problems introduced by
 %  
  %We consider different  constructions of $D$, one based  on a permutation of $\{1,\ldots,n\}$ obtained by running the (NN) algorithm 
  %5
  %using  the metric from   \cite{zhang2015estimating}. The permutation arises  from  performing the nearest neighbor (NN) algorithm which is an approximate solution to  the traveling %salesman problem \cite{rosenkrantz1977analysis}. Hence, we refer to the respective  estimator as NN-FL.
   %We also analyze the performance of the estimator obtained by  using the permutation that comes  from sorting the degrees of the nodes. The resulting estimator   is similar in spirit %to the method from \cite{chan2014consistent}, but as we can see in (\ref{eqn:sorting}), we use a different penalty based on $\ell_1$  regularization. To distinguish it from  the %previous estimator,  our second estimator is referred to as

\subsection{Summary  of results}

As stated  above, our approach to  graphon estimation is based on performing  fused lasso/total variation denoising  of the adjacency matrix  with  appropriate graphs.  Loosely speaking, for the resulting estimators  we show the following:

\begin{enumerate}
	
		\item If the function $f_0(u,\cdot)$  is piecewise H\"{o}lder  of exponent  $\alpha \in  (0,1/2]$, then the NN-FL estimator  attains the rate $n^{-\alpha/2}$   after disregarding    logarithmic terms.
		In fact, our  result  actually  holds  for  a class of  functions larger than that of piecewise H\"{o}lder functions  of exponent  $\alpha$.

	\item  Let $g$  be the degree function  $g(u)  \,=\,  \int f_0(u,v)\,dv   $.  If  there exists a constant $L \,>\, 0$  such that  $L\,\vert u\,-\, v\vert\,\leq \, \left\vert g(u)\,-\, g(v) \right\vert $  for all $u,v \in [0,1]$,  $g$  is piecewise--monotone,   and $f_0$  has bounded variation, then  the  mean  square error (MSE) of the SAS-FL    attains rate  $n^{-1/2}$ (if we ignore logarithmic terms).
	
	%the class for which this result is shown to hold is much larger than that of piecewise H\"{o}lder functions  of exponent  $\alpha$.
	%We show that if $f_0$  is  of bounded variation,
	
%	\item We show that both NN-FL and SAS-FL  attain the  rate $K\,\log n/n$  when the true generative  model is the stochastic block model with $K$  communities.  This is similar to the performance  of USVT which attains the rate $K/n$, as shown in \cite{xu2017rates}. A minor difference is  that the rate in \cite{xu2017rates}  is on the expected mean squared error,  whereas our upper bounds are concentration inequalities  for the MSE. 

\item  In both  simulated  and real  data examples,  we provide  evidence  that the proposed methods  outperform  existing approaches  for    network  denoising and link prediction.
\end{enumerate}

\subsection{Notation}

For  $n \in \mathbb{N}$  we denote   the set $\{1,\ldots,n\}$    by $[n]$.  Moreover,  for $n \in \mathbb{N}$  we denote  by $\mathcal{S}_n$  the set of  permutations  of $[n]$, thus bijections from $[n]$ to itself.

  For  a matrix  $X \in \mathbb{R}^{s \times t}$,     the $i$-th  row of  $X$ is  denoted by  $X_{i,\cdot}$ , and   the $i$-th column by   $X_{\cdot,i}$. The Frobenius norm is denoted  by 
$$  \|X\|_F \,\,=\,\,  \sqrt{ \displaystyle \sum_{i=1}^{s}\,\sum_{j=1}^{t}  x_{i,j}}.   $$
For a set $I \subset  [t] $ we denote by $X_{\cdot,I}$  the matrix  obtained   by removing  from $X$  the columns  whose indices correspond to elements in $[t] \backslash I$. We define  $X_{I,\cdot}$  in a similar way for $I \subset [s]$.  

Throughout, we use the standard notation $\| x \|_p =  ( \sum_{i=1}^{n} \vert x_i \vert^p)^{1/p}$  for  $x \in \mathbb{R}^n$, and $p \in \mathbb{N}$ along with the  
 $\|x\|_{\infty }  =  \max_{i \in [n]} \vert x_i\vert$  notation.

For  a  Borel measurable set $A \subset \mathbb{R}^d$ we denote  its  Lebesgue  measure in $\mathbb{R}^d$ by $\text{Vol}(A)$.  And we denote by $1_A(x)$  the function that takes value $1$ if  $x \in A$  and $0$  otherwise.

\subsection{Previous work}

Methods  for  network estimation have been extensively  studied in the literature,  and network denoising remains  an active research area due to the challenge that it represents.
One avenue of research has focused on assuming that the generative process is the stochastic block model. In such framework, the main difficulty is perhaps to estimate the communities to which the nodes belong.
This perspective includes  seminal works by 
\cite{bickel2009nonparametric,rohe2011spectral}  and \cite{adamic2005political}, more recently by \cite{guedon2016community, yan2018provable,gao2018minimax}  among others.

In this paper we will not necessarily assume the stochastic block model but will also allow  for more general models. Our work focuses on directly   estimating the link probabilities  to which we refer as  graphon estimation. The statistical properties  of graphon estimators  have been of extensive interest.  For instance,
\cite{airoldi2013stochastic}  proposed to  approximate  Lipschitz  graphons  by  block models, and consistency in  H\"{o}lder classes  was studied in     
\cite{wolfe2013nonparametric}. Moreover,   \cite{gao2015rate}  and \cite{gao2016optimal}  characterized the   minimax rates for graphon estimation under the stochastic block model.  The rate is $\log K /n$  where $K$, the number of  communities, satisfies   $K  \,\leq \, \sqrt{n \log n}$. Recently, \cite{xu2017rates}  and \cite{klopp2017optimal}  independently  showed that 
the universal singular value thresholding (USVT) algorithm from \cite{chatterjee2015matrix} attains the rate $K/n$   when the true model is the stochastic block model. In fact, 
\cite{xu2017rates}  showed the rate $K/(n\,\rho)$  where $\rho $  is a sparsity parameter that satisfies $n\,\rho \,=\,  \Omega(\log n)$.  There, the authors also studied  USVT  when the function $f_0$ (in Model \ref{eqn:model1}) belongs to a Sobolev space.

%chatterjee2015matrix
Moreover,   \cite{zhang2015estimating}  studied  graphon estimation under more general structures than  the stochastic block model.
 Their method  proceeds in a two--step fashion. First, a network with nodes $\{1,\ldots,n\}$  is constructed  using the adjacency matrix $A$. Then, neighborhood smoothing is applied to the different neighborhoods. The resulting  estimator is proven to be consistent when $f_0 $  is piecewise Lipschitz.  The respective  rate is $(\log n /n)^{1/2} $  versus  the minimax rate $\log n /n$. A different method based on neighborhood smoothing was studied in \cite{song2016blind}.
 %A different method   was studied in  \cite{song2016blind}  that  is also based on neighborhood smoothing. 

An alternative approach to graphon estimation considering degree sorting was studied in \cite{chan2014consistent}     under the  assumption that $\int_{0}^{1}  f_0(x,y)dy$  is increasing. The idea  behind this method is to construct  an ordering of the nodes based on the empirical degree $d_i \,=\,  \sum_{j=1}^{n} A_{i,j}$. Once the ordering is obtained, the authors  propose the sort and smooth (SAS) estimator
\begin{equation}
\label{eqn:sorting}
\hat{\theta}_{sas} \,\,=\,\,   \underset{\theta}{\arg\min}  \displaystyle  \sum_{i=1}^{m} \sum_{j=1}^m  \,\sqrt{    \left( \frac{\partial \theta}{\partial x}  \right)_{i,j}^2  \,\,+\,\,\left( \frac{\partial \theta}{\partial y}\right)_{i,j}^2    }    \,\,\,\,\,\, \text{subject to }   \,\|  \theta  - H   \|_2  \,\leq \,\epsilon, 
\end{equation}
%chan2014consistent
where $\epsilon > 0$  is  a tuning parameter, $H$   is a smooth version of $A$,  and $ \frac{\partial \theta}{\partial x}$  $ \left(\frac{\partial \theta}{\partial y}\right)$  is a discrete  derivative in the direction of the x (y) axis.

Finally,  we  emphasize  that while several  graphon estimation methods exist, many properties of estimators  are unknown. For instance,  the estimator from \cite{zhang2015estimating}  is consistent  when the true graphon is  piecewise Lipschitz  but it is not known if such estimator can  be nearly minimax when the true  model is the stochastic block model. Also,  USVT  can  perform well under the stochastic block model assumption or some smooth condition on $f_0$,  but  performance guarantees are unknown when  $f_0$ is  piecewise H\"{o}lder  or it is a general bounded variation function.% In fact, if the $f_0$ is of bounded variation and not necessarily piecewise Lipschitz,  it is unknown  if there is even   polynomial time method that can consistently estimate $\theta^*$.  

\section{General approach}

In this section  we  propose  a general class of estimators  filling in the details of our discussion in Section \ref{sec:intro}. We begin by    reviewing some background on the traveling salesman problem. This is then related to the graphon estimation problem giving rise to our family of estimators. 

\subsection{Nearest  neighborhoods  construction}
\label{sec:nn}

Our approach to construct a grid graph is motivated by the traveling salesman problem.
We describe this in generality next. Let $\mathcal{C} \,\,=\,\, \{c_1,\ldots,c_s\}$ be a set of cities with a  distance  metric $d_{\mathcal{C}} \,\,:\,\,\mathcal{C} \times \mathcal{C} \,\rightarrow \mathbb{R}$ specifying how far any pair of cities are from each other. A tour of cities is a  sequence $c_{P(1)},\ldots,c_{P(s)}$ where $P$ is a permutation of the set  $\{1,\ldots,s\}$. Thus, a tour is just an arrangement of cities such that each city appears exactly once. %Throughout, we denote by  $\mathcal{S}_{s}$  the set of all permutations acting in the set  $\{1,\ldots,s\}$.

% $P \in \mathcal{S}_{K}$ is the permutation group acting on $\{1,\ldots,K\}$. Thus, a tour is just an arrangement of cities such that each city appears exactly once.
The traveling salesman problem  (see \cite{bellmore1968traveling} for a survey) is a well known problem that seeks for the tour with the minimum cost, measuring the cost  in terms of the metric 
$d_{\mathcal{C}}$, see below. The optimal tour or circuit is found by solving  
\begin{equation}
\label{tsp}
P^*\,\,\in \,\, \underset{P   \in \mathcal{S}_s }{\arg \min}\,\, C(P),
\end{equation}
where the cost of a tour $P$ is defined as
\[
C(P)\,\,\,=\,\,\,\sum_{i=1}^{s-1}  \,d_{\mathcal{C}}(c_{P(i)},c_{P(i+1)})\,\,+\,\,  d_{\mathcal{C}}(c_{P(s)},c_{P(1)}) .
%C(\hat{P} )\,\,\leq \,\, C(P^*).
\]
Unfortunately, it is known that (\ref{tsp}) is NP-hard, e.g.  \cite{rosenkrantz1977analysis}. Despite this challenge, there exist different approximation  algorithms for solving (\ref{tsp}). For instance,  \cite{rosenkrantz1977analysis} studied approximation algorithms which run in polynomial time. From such methods,  we will use the nearest neighbor algorithm (NN) which starts from a random city, and then proceeds iteratively by choosing the closest city to the current city. Specifically,  at first the algorithm visits $\hat{\tau}(1)$  for some $\hat{\tau}(1) \in \mathcal{C}$, perhaps chosen  randomly. Next,  the NN algorithm visits $\hat{\tau}(2)$, where  
\[
\hat{d}(\hat{\tau}(1),\hat{\tau}(2)) \,\,\leq \,\,    \hat{d}(\hat{\tau}(1),c),\,\,\,\,\,\, \forall  c \in \mathcal{C} \backslash \{ \hat{\tau}(1) \} .
\]
Then, the NN algorithm continues sequentially visiting  $\hat{\tau}(j)$   at iteration $j$  where
\[
\hat{d}(\hat{\tau}(j-1),\hat{\tau}(j)) \,\,\leq \,\,    \hat{d}(\hat{\tau}(j-1),c),\,\,\,\,\,\, \forall  c \in \mathcal{C} \backslash \{\hat{\tau}(1),\ldots,\hat{\tau}(j-1)  \}.
\]

Although, the nearest neighbor method is straightforward, it enjoys some attractive performance guarantees. In particular, 
denoting the permutation associated with NN by $\hat{P}$, Theorem 1 in  \cite{rosenkrantz1977analysis} shows that
\begin{equation}
\label{eqn:nn_property}
C(\hat{P} )\,\,\leq \,\, \left(1 + \frac{\log_2 s}{2} \right)\,\, C(P^*),
\end{equation}
provided  that  $\hat{d}$  satisfies  the triangle inequality.

Moreover, by its mere definition, one can see that the computational complexity of running the NN algorithm  is $O(w\,s^2)$, where $w$ is the computational cost of evaluating $d_{\mathcal{C}}$.\\

\subsection{Class of estimators}
\label{sec:class_of_estimators}

As anticipated,  our approach is based on running   2D  grid fused lasso denoising on the data $A$  with a carefully constructed graph.  We will now state more precisely  how to arrive to our estimator  $\hat{\theta}  \in  \mathbb{R}^{n \times n} $  for   $\theta^*$.

In order to avoid correlation  between  the constructed ordering  and the signal to denoise,  for  $m  \in [n]$ we   use the data $A_{\cdot,[m]}$  to construct 
metric  $\hat{d}  \,:\, ([n]  \backslash [m] )\times ([n]  \backslash [m] )  \rightarrow \mathbb{R}$. Once the metric  is constructed, we can think of the elements of $([n] \backslash  [m])$  as a set of cities,   and for any two cities $i,j \in [n]$,  the distance $\hat{d}(i,j)$  tells us how far the cities  are. A  discussion on  choices of the metric  $\hat{d}$  will be given later.

 \textbf{Choice of $m$:} Throughout  we assume that  $m $ is chosen  to satisfy $m \asymp n$, thus there exists positive  constants $c_1$ and $c_2$  such that $c_1\,m \,\leq \, n  \,\leq  \,c_2\,m$. For instance, we could take  $ m  =  \floor{n/2}$.

A natural  way to arrange the cities (nodes), in the graphon estimation problem,  would be to  place   cities  that are close to each other in the sense of the metric  $\hat{d}$ as adjacent. This would make sense  if the  graphon  has some  underlying structure, for instance  if the ground truth is the stochastic block model. We would also  require that  the distance $\hat{d}(i,j)$  is a reasonable  approximation to a metric $d^*(\xi_i,\xi_j)$,  such that  $f_0(\xi_{i},\cdot)$  and $f_0(\xi_j,\cdot) $  are ``similar" if  $d^*(\xi_i,\xi_j)$  is small. We  will be precise  in  stating our assumptions but for now we proceed to construct our proposed approach.

Motivated  by the  discussion above,  we use the NN  algorithm  (as  discussed in Section \ref{sec:nn}) on the cities  $[n] \backslash [m]$   with distance  $\hat{d}$. We let    $\hat{\tau}$   be the corresponding  function  $\hat{\tau} \,:\,  [n-m] \,\rightarrow\,   ([n] \backslash [m]) $, such  that the NN algorithm first visits city $\hat{\tau}(1)$, next city $\hat{\tau}(2)$,  and so on.

% the NN algorithm first visits city $\hat{\tau}(1)$, next city $\hat{\tau}(2)$, where  $\hat{\tau}(2)$ is the closest point  among $([n] \backslash [m]) \backslash  \{\hat{\tau}(1)\}$ to  $\hat{\tau}(1)$, and so on. 
Using the ordering $\hat{\tau}$,  we construct  a signal    $y \in \mathbb{R}^{ (n-m) \times (n-m)    }$    satisfying  $y_{i,j} \,=\,   A_{ \hat{\tau}(i),  \hat{\tau}(j)   }  $ for all $i,j \in [n-m]$.  We also construct the 2D  grid graph  $G \,=\, (V,E) $  with  set of nodes%$(n-m) \times  (n-m) $  nodes. The  set of  
\[
   V  \,=\,  \{(i,j)  \,:\,   i,j \in [n-m]  \},
\]
and set of edges   
$$E  \,=\, \{ (e^+,e^-)  \in [n-m]^2   \times [n-m]^2  \,:\,  \|  e^+  \,-\,e^- \|_1  = 1     \}.   $$
 We also use  $\nabla_G$ to denote    an oriented  incidence  operator of $G$. Thus,  $\nabla_G   \,:\,  \mathbb{R}^{(n-m)\times (n-m)}  \,\rightarrow\,  \mathbb{R}^{\vert E \vert} $  where    for  $e \,=\,  (e^+,e^{-}) \in E$  we have
 \[ 
 (\nabla_G \theta)_{e}  \,\,=\,\, \theta_{e^+} \,\,-\,\, \theta_{e^-}.
 \]
 
  Using the graph   $G$,  we proceed   to construct  our estimator  by  solving  a graph  fused lasso  problem. Doing so, we first  find  $\hat{\beta}  \in \mathbb{R}^{ (n-m) \times (n-m) }$
  as the solution to
  \begin{equation}
  \label{eqn:estimator}
  \underset{ \beta \in \mathbb{R}^{ (n-m) \times (n-m) } }{ \text{minimize} }  \,\,  \frac{1}{2} \| y \,-\,  \beta \|_F^2\,\,+\,\,  \lambda\,\|\nabla_G \beta \|_1,    
  \end{equation}
  for a tunning parameter $\lambda  >0$. We then set  $\hat{\theta}_{  \hat{\tau}(i),\hat{\tau}(j)}   \,=\,    \hat{\beta}_{i,j}$  for  all  $i,j \in [n-m]$. 
  
  The procedure above  allows us to estimate  $\theta^*_{i,j}$ for  all  $i,j \in [n]\backslash [m]$. In a similar way we can also construct  estimates  of  $\theta^*_{i,j}$ for  all  $i,j \in  [m]$.  The idea is to  have  an ordering of the nodes  in $[m]$  by  using the NN  algorithm with a metric that only involves data from $A_{ \cdot,  ([n] \backslash  [m]) }$, and then  we   solve  a 2D fused lasso  problem.
  
%  Hence  this produces  estimates of  $\theta^*_{i,j}$ for  all  $i,j \in [n]\backslash [m]$. With a similar procedure we can also construct  estimates  of  $\theta^*_{i,j}$ for  all  $i,j \in  [m]$.
  
 As for  the estimates of $\theta^*_{i,j}$ for  all  $i \in [n]\backslash [m]$  and  $j \in [m]$,  we proceed  in a similar way but using two different orderings. The first ordering   $\hat{\tau}_1$  is obtained  by running the NN algorithm on the set of cities $[m]$ with a metric  depending on the data $A_{[m], [m] }$.  For the second ordering  we  use the NN algorithm  with set of cities $([n]\backslash [m])$  and metric depending the data $A_{([n]\backslash [m]),([n]\backslash [m])}$. 
    %They are both based on first constructing a  metric  $\hat{d}$  in the set $[n]$, with  $\hat{d}$ only using the data    $A_{\cdot,[n] \backslash [m] }$. The metric $\hat{d}$  then allows us to consider the sets of cities $[m]$  and  $[n] \backslash [m]$  and to  use the NN algorithm to obtain tours  $\hat{\tau}_1$  and $\hat{\tau}_2$  respectively. This then  induces the matrix  $\tilde{y} \in \mathbb{R}^{(n-m) \times m}  $  as  $\tilde{y}_{i,j}  \,\,=\,\,A_{  \hat{\tau}_1(i), \hat{\tau}_2(j) }$. The latter  signal can be then denoised with the fused lasso on a  grid  graph  to produce the desired estimates.
    
    Finally,  we emphasize that  we do the portioning of the data in order to avoid correlation between the  ordering and the signal  to denoise. This is done to keep the  analysis  mathematically correct. However, in practice, one can obtain  a single  ordering  and then run a  fused lasso problem  as in (\ref{eqn:first_problem}).  We have noticed that such  an  approach  works well in practice.

 \paragraph{Computational cost.}
  As stated in the previous subsection, the computational cost  associated with the NN  algorithm is $O(s\,n^2)$, where $s$  is the cost of computing   the distance between any two  nodes.
Note that this  can be reduced  if there are  multiple processors available.  In such case, the distance  from a node to the remaining nodes  could be computed by partitioning the  nodes and performing computations in parallel. 

As for the fused lasso computation, this can be done using the efficient algorithm  from \cite{barbero2014modular},  or the  ADMM  solver from \cite{tansey2015fast}  which is amenable to parallel computing. 
  
\subsection{Choices of  metric $\hat{d}$}
\label{sec:d}

Clearly,   the class of methods  described above  can be used  with any metric  $\hat{d}$  on the set of nodes on the graphon  estimation. The purpose  of this  section is to highlight  different  choices  of $\hat{d}$, some of which  have  appeared in the literature in the context of other  estimators.

%two  metrics  that have  appeared in the literature in the context of other  estimators.

For simplicity,  we will focus on constructing $\hat{d}$  for  the  case of estimating  $\theta^*_{i,j}$ for  all  $i,j \in [n]\backslash [m]$. The  remaining  cases  described in 
Section \ref{sec:class_of_estimators}  can be  constructed in a similar way.

\subsubsection{Inner product based  distance}
%\paragraph{.}
Our first natural approach for constructing  $\hat{d}$
 is to  consider the  metric proposed  in \cite{zhang2015estimating}. Specifically,  we set
%Our first approach for constructing  $\hat{d}$  is based on the distance studied   in \cite{zhang2015estimating}. Specifically,  we set
\begin{equation}
\label{eqn:distance_smoothing}
\hat{d}_I(i,i^{\prime})  \,\,=\,\, \,  \underset{  k   \in [m] }{\max }\, \sqrt{\frac{1}{n} \vert  \langle  A_{i, [m]},  A_{k, [m] } \rangle    \,-\,    \langle  A_{i^{\prime}, [m]},  A_{k, [m] } \rangle   \vert      }, \,\,\,\,\,\forall  i, i^{\prime}  \in [n] \backslash [m].
\end{equation}
We  call NN-FL  the   estimator   from Section \ref{sec:class_of_estimators}  when  the distance   $\hat{d}$ is taken as $\hat{d}_I$ in (\ref{eqn:distance_smoothing}).

Importantly, our  modification of the  distance  defined in  \cite{zhang2015estimating}  does  satisfy the  triangle inequality,  which is sufficient for the NN algorithm  to satisfy 
(\ref{eqn:nn_property}).

%We will focus on constructing $\hat{d}$  for  the  case of estimating  $\theta^*_{i,j}$ for  all  $i,j \in [n]\backslash [m]$. 
%\paragraph
\subsubsection{Sorting}
 
  Another  choice for the  distance $\hat{d}$ is
%The other natural  choice for the  distance $\hat{d}$ is to  consider
\begin{equation}
\label{eqn:distance_sorting}
\hat{d}_1(i,i^{\prime})  \,\,=\,\,    \left\vert  \frac{1}{m} \sum_{j \in [m]}A_{i,j}   \,-\,  \frac{1}{m} \sum_{j \in [m]}A_{i^{\prime},j} \right\vert,  \,\,\,\,\forall i,i^{\prime }  \in [n] \backslash [m]. 
\end{equation}
Thus,  for  any two  nodes, the distance is the absolute  value of the difference  between  the degrees (normalized by $m$) based on  the data  $A_{\cdot,[m]}$. Since  such degrees  are  numbers and the metric is the Euclidean  distance,  the optimal  tour (traveling salesman problem solution ) is the ordering obtained by sorting the degrees. This is  the  ordering  constructed in the  method from \cite{chan2014consistent}. The difference  is that we use the fused lasso penalty for denoising without preliminary smoothing of the data, whereas \cite{chan2014consistent} use the penalty in (\ref{eqn:sorting}) with a smoothed version of $A$.

%However,  we emphasize an important difference  between such approach and ours. This  difference consists of the penalty used for  denoising. We use total variation   in the sense  of fused lasso regularization, whereas \cite{chan2014consistent} uses  the penalty in (\ref{eqn:sorting}).  We refer the reader  to \cite{sadhanala2016total}, for an example where   using the fused lasso  penalty  can lead to minimax estimators in contrast to other commonly used penalties.

Throughout, whenever  we use the  distance  (\ref{eqn:distance_sorting})  we will refer to the estimator  from Section \ref{sec:class_of_estimators}  as sort and smooth fused lasso  graphon estimation (SAS-FL). In the experiments section we will see that, as expected, the empirical performance of SAS-FL  is similar to  SAS  from \cite{chan2014consistent}, although the former seems to perform sightly  better in the examples in Section \ref{sec:experiments}.

%In addition, 
%

%in this case a tour (not all of them)  of the NN  approximation to the TSP  consists of sorting the  nodes  da%ta according to the  degree $  D_i \,:=\,\sum_{j=1}^{n} A_{i,j}$.
%Hence, our approach is similar to the method from \cite{chan2014consistent}  with the important difference that we use the fused lasso penalty.

  %\paragraph
  \subsubsection{$\ell_1$  distance and other choices}
% \section{Theory}
 
% The purpose of this section is  to study  convergence  proprieties  of the NN-FL  and SAS-FL  estimators. We organize  our results  by  analyzing each method one at a time. 
 Clearly,  a different  metric  can be obtained   by simply taking the  $\ell_1$  norm  of the difference between rows  or columns  of the incidence  matrix $A$.  Thus, we can define 
 \[
   d_{\ell_1}(i,i^{\prime})  \,=\,  \| A_{i,[m]} \,-\, A_{i^{\prime},[m]}  \|_1,  \,\,\,\,\,\forall  i, i^{\prime}  \in [n] \backslash [m].
 \]
We will refer to the  respective procedure using this metric as  $\ell_1$-FL.  We will not study  convergence properties of this method, although  we will see that it is a reasonable method in practice.

%\paragraph{External metrics.}

Finally,  we  notice that  the  metric  $\hat{d}$  could  be constructed using side information about the nodes. For instance, using covariates, or repeated measurements  of the network if available. 

 \section{Analysis of the NN-FL  estimator on extensions of piecewise H\"{o}lder classes}
 \label{sec:extenshions_holder}
  
  The purpose of this section is  to study  convergence  properties  of the  NN-FL  estimator. 
 We analyze  the  performance of the NN-FL estimator for classes of graphons   that  extent  the notion of piecewise H\"{o}lder functions. We notice that a particular  instance of this  was studied in \cite{zhang2015estimating}. 
 
 To formalize, if  $\alpha \in (0,1]$, we say that $f_0$  is piecewise H\"{o}lder of exponent $\alpha$  if the following holds. There exists a  partition of intervals $\mathcal{A}_1,\ldots,\mathcal{A}_r$  of $[0,1]$, such that if $u,v \in \mathcal{A}_l$  for some $l \in [r]$, then
 \begin{equation}
 \label{eqn:piecewise_holder}
 %	 \[
 \underset{t \in [0,1]}{\sup}\, \left\vert f_0(u,t) \,-\,  f_0(v,t)  \right \vert\,\,\leq \,\, L_1\,\vert u\,-\, v\vert^{\alpha}, 
 % 	 \] 
 \end{equation}
 and 
 \begin{equation}
 \label{eqn:piecewise_holder2}
 %	 \[
 \underset{t \in [0,1]}{\sup}\, \left\vert f_0(t,u) \,-\,  f_0(t,v)  \right \vert\,\,\leq \,\, L_1\,\vert u\,-\, v\vert^{\alpha}, 
 % 	 \] 
 \end{equation}
 for  a positive  constant  $L_1$  independent of  $u$ and $v$.
 %$a_0 = 1 \, <\,   a_1   \,<\, \ldots  \,<\,  a_m =1  \,$   such that   if $u,v \in $ 

 This class  of functions  appeared  in the   analysis of the fused lasso estimator in the context of 2D non-parametric  regression, see \cite{hutter2016optimal}. There,  the authors  showed that  the fused lasso  attains  the rate  $n^{- \alpha/(1+\alpha)   }$.
 
 Next, we state  two assumptions  which  hold if (\ref{eqn:piecewise_holder}) is satisfied along with Model  \ref{eqn:model1}. Thus, we relax  the piecewise H\"{o}lder  condition. 
 
 \begin{assumption}
 	\label{as1}
 	There exists a positive constant $c_1 $	  such that with probability   approaching  one 
 	\[
 	\underset{  k \in [m]  }{\min}   \displaystyle  \int_{0}^1   \vert  f_0(\xi_k,t)  \,-\,  f_0(\xi_{i},t)  \vert dt \,\leq \,\,    c_1  \left(  \frac{\log n}{n}\right)^{\alpha} ,\,\,\,\,\,\, \forall  i  \in  [n] \backslash  [m].      
 	\]
 	and
 	\[
 	\underset{  k \in [m]  }{\min}   \displaystyle  \int_{0}^1   \vert  f_0(t,\xi_k)  \,-\,  f_0(t,\xi_{i})  \vert dt \,\leq \,\,    c_1  \left(  \frac{\log n}{n}\right)^{\alpha} ,\,\,\,\,\,\, \forall  i  \in  [n] \backslash  [m].      
 	\]
 	for $\alpha \in (0,1/2]$.
 \end{assumption}
 
 Note that we constrain $\alpha \in (0,1/2]$,   and the case $\alpha \in [1/2,1]$  will not be studied. Instead,  we will analyze  the stochastic  block model in the next subsection. %The next section  will also explore the more general case of bounded variation.
 
 To understand why Assumption  \ref{as1}  holds  when  condition (\ref{eqn:piecewise_holder})  is met with $\alpha \in (0,1/2]$, we refer to the  work in \cite{von2010hitting}.  There, the authors showed that, with probability approaching one,  the following holds: For each $\xi_i$, the  set of its $K $-th nearest neighbors (with Euclidean distance)  among the points $\{x_j\}_{j \in [n ] \backslash \{i\} }$  are all within distance $c\log n /n$  for an appropriate  choice of $K$, and for some constant $c >0$.

 We now state our second assumption. This involves  the  quantity  that the penalty $\|\nabla_G  \cdot\|_1$ in (\ref{eqn:estimator})  emulates. % The scaling that  we use for such modelling total variation is the one that would arise  in  the case  of a piecewise H\"{o}lder   graphon.        
 
 \begin{assumption}
 	\label{as2}
 	There exists   an unknown (random) permutation $\tau^* \in \mathcal{S}_{n-m-1}$  such that
 	\[
 	\displaystyle \sum_{i=1}^{n-m-1}  \displaystyle \int_0^1    \vert f_0(\xi_{ \tau^*(i) },t) \,-\,   f_0(\xi_{\tau^*(i+1) },t)  \vert   \,dt \,\,=\,\,  O_{\mathbb{P}}(  n^{1-\alpha} \,\log^{\alpha} n   ),
 	\] 
 	and
 	\[
 	\displaystyle \sum_{i=1}^{n-m-1}  \displaystyle \int_0^1    \vert f_0(t,\xi_{ \tau^*(i) }) \,-\,   f_0(t,\xi_{\tau^*(i+1) })  \vert   \,dt \,\,=\,\,  O_{\mathbb{P}}(  n^{1-\alpha} \,\log^{\alpha} n   ),
 	\] 
 	where $\alpha \in (0,1/2]$.
 \end{assumption}
 
 If (\ref{eqn:piecewise_holder}) and (\ref{eqn:piecewise_holder2}) hold then Assumption \ref{as2}  will  be satisfied. To verify this,  simply take $\tau^*$  as the ordering of the elements of $\{\xi_i \}_{i \in  ([n ]\backslash [m])}$, i.e,
 \[
 \xi_{\tau^*(1)} \,<\,\xi_{\tau^*(2)}  \,<\,  \ldots \,<\, \xi_{\tau^*(n-m)}.
 \]
 Once again, we exploit properties of nearest neighbor graphs from \cite{von2010hitting}. %.   as  with probability  approaching one,  for   

 With these conditions, we are now ready  to present our next  result.
 
 \begin{theorem}
 	%	\label{thm:bound_on_tv}
 	%	Let us suppose that Assumptions \ref{as1}-\ref{as2}  hold. Then,  we  extend $\theta^*$  to be in $\mathbb{R}^{ n \times n}$, by setting $\theta*_{j,i}  \,=\,   \theta^*_{i,j}$  for all
 	%	$i < j$  with  $i,j \in [n] \backslash [m]$. Moreover,  we set  $\theta_{i,i}^*   \,= \,  0$   for all  $i \in [n] \backslash [m]$. Let $\hat{\tau}$ be constructed  as in Section \ref{sec:class_of_estimators}  by setting $d \,:=\, \hat{d}_I$. Then,
 	%	%using the metric $\hat{d}_I$  as explained in 
 	%	\[
 	%	\displaystyle    \frac{1}{n^2}   \sum_{i=  1}^{n-m  - 1}  \sum_{j=m+1}^{n} \vert  \theta_{ \hat{\tau}(i),j}^* \,-\, \theta_{\hat{\tau}(i+1),j}^*\vert    \,\,=\,\, O_{\mathbb{P}}\left(     \frac{  \log^{\frac{(1+\alpha)}{2}   }n  }{n^{  \frac{\alpha}{2} }}  \right).
 	%	\]
 	\label{thm:rate}
 	Let us suppose that Assumptions \ref{as1}-\ref{as2}  hold, and  let $\hat{\tau}$ be constructed  as in Section \ref{sec:class_of_estimators}  by setting $d \,:=\, \hat{d}_I$ (see Equation \ref{eqn:distance_smoothing}). Then  for an appropriate choice of $\lambda$,  the corresponding estimator  in (\ref{eqn:estimator})   satisfies
 	%using the metric $\hat{d}_I$  as explained in 
 	\[
 	\displaystyle    \frac{1}{n^2}   \sum_{ i,j \in [n-m],  \,i < j  }  \left( \hat{\theta}_{ \hat{\tau}(i), \hat{\tau}(j)   }  \,-\,   \theta_{\hat{\tau}(i), \hat{\tau}(j)  }^*  \right)^2   \,\,=\,\,  
 	O_{\mathbb{P}} \left(   \frac{  \log^{2\,+\, \frac{\alpha}{2}   }n  }{n^{  \frac{\alpha}{2} }}  \right).
 	% \vert  \theta_{ \hat{\tau}(i),j}^* \,-\, \theta_{\hat{\tau}(i+1),j}^*\vert    \,\,=\,\, O_{\mathbb{P}}\left(     \frac{  \log^{\frac{(1+\alpha)}{2}   }n  }{n^{  \frac{\alpha}{2} }}  \right).
 	\]
 	%Assume that Assumptions \ref{as1}-\ref{as2}  hold. Then,  we  extend $\theta^*$  to be in $\mathbb{R}^{ n \times n}$, by setting $\theta*_{j,i}  \,=\,   \theta^*_{i,j}$  for all
 	%	$i < j$  with  $i,j \in [n] \backslash [m]$. Moreover,  we set  $\theta_{i,i}^*   \,= \,  0$   for all  $i \in [n] \backslash [m]$. Let $\hat{\tau}$ be constructed  as in Section \ref{sec:general_estimators}  by setting $d \,:=\, \hat{d}_I$. Then,
 	%	%using the metric $\ha
 \end{theorem}

 Thus,  in the  class of  graphons implied by Assumptions \ref{as1}--\ref{as2}, the NN-FL estimator  attains the rate $n^{-\alpha/2}$  after ignoring  logarithmic terms. To the best of our knowledge, other estimators have not been studied  on this class of functions. The most related  work comes from \cite{zhang2015estimating}  who studied  piecewise Lipschitz graphons  for which their estimator attains the rate $n^{-1/2}$.

 % \subsection{BV-bi-Lipschitz classes}
 % \label{sec:BV-bi-Lipschitz classes}

%s.
  
  \section{Analysis of the SAS-FL estimator on BV  functions}

 % \subsection{BV  functions}
  %\label{sec:bv_class}
  
 We focus on the class of  graphons that satisfy  the  assumption  of bounded variation. This condition  has appeared  in the statistics literature  due to its flexibility. Indeed, it encompasses a very large  class of functions that includes reasonable subclasses such as piecewise Lipschitz. In one dimensional  non-parametric  regression,  \cite{mammen1997locally}  studied   locally adaptive  estimators    that attain minimax rates when the true  regression function  has bounded variation. More recently, \cite{tibshirani2011solution,tibshirani2014adaptive}   studied  discretized  version of the estimators from \cite{mammen1997locally}. The framework   from Section \ref{sec:class_of_estimators} consists of  a particular  instance of  generalized lasso estimation studied in \cite{tibshirani2011solution}. 
 
 In one dimension, bounded variation is well defined  as follows.  A function  $f\,:\, [0,1] \,\rightarrow \,\mathbb{R}$   has bounded variation if there exists a constant $C$  such that  for any set of points    $0 \,\leq\, a_1  \,\leq \,\ldots\,\leq  a_r  \leq  1$,  $r \in \mathbb{N}$, it holds that
\begin{equation}
	\label{eqn:one_d_bv}
%	\[
	\displaystyle  \sum_{l=1}^{r-1} \, \vert  f(a_l) \,-\,f(a_{l+1}) \vert     \,\,<\,\,C.
%	\]
\end{equation}

When passing to  higher dimensions (in particular dimension two),  the definition  of bounded  variation  is not unique. An early work  from  \cite{clarkson1933definitions}  discussed multiple  definitions of bounded variation. Although these   days there is a widespread convention  in the definition of bounded variation in the field of  mathematics, the statistics community continues  to rely on early  definitions. For instance, perhaps implicitly \cite{sadhanala2016total}   defined  the canonical  class of bounded variation by taking  (\ref{eqn:one_d_bv})   and  assuming it holds  through  each  horizontal  and vertical chain  graph of a 2D grid graph. We now do something similar for the case of graphons.   
%In the machine  learning  community, \cite{sadhanala2016total}   defines  the canonical bu taking (\ref{eqn:one_d_bv})   and  assuming it holds  through  each  horizontal  and vertical chain  graph of a 2D grid graph. We now do something similar for the case of graphons.   

 % We now define  the class of functions over which  we analyze  the SAS-FL estimator.

  \begin{assumption}
  	\label{as5}
  	We assume the data is generated as in the model implied by  (\ref{eqn:model1}). Moreover, we assume that the functions %$f_0$  satisfies  $f_0(x,y) \,=\, f_0(y,x)$  and that
  	\[
  	g_1(u) \,\,=\,\,  \displaystyle \int_0^1\, f_0(u,v)dv,\,\,\,\,\,\,\,\,\,\,\,\,\,	g_2(u) \,\,=\,\,  \displaystyle \int_0^1\, f_0(v,u)dv,
  	\]
  	satisfy  the  following:
  	\begin{enumerate}
  		\item   There  exists  some positive constant $L_1$  such that 
  		\begin{equation}
  			\label{eqn:bilip}
  				L_1  \,  \vert  x \,-\, y    \vert  \,\,\leq \,\,   \vert  g_l(x) \,-\,g_l(y) \vert , \,\,\,\,\,\,\,\,\forall  x,y \in  [0,1],  \forall  l \in  \{1,2\}.  				
  		\end{equation}
  		% for some  $\alpha \in (1/2,1]$.
  		\item \textbf{Piecewise-Monotonic:} For  $l\in \{1,2\}$ there exists a partition  $0\,<\,b_1^l\,<  \ldots\, <\,  b_r^l   \,<\,  1$  such that $g_l$  is  monotone in each of the intervals 
  		$(0,b_1^l),  (b_1^l,b_2^l),\ldots, (b_r^l,1)$.
  		
  		\item  The function $f_0$ has \textbf{bounded variation} in the following sense. There exists a positive constant $C>0$  such that  if  $ 0\,\leq\,  a_0  \,\leq \, a_1 \,\leq \,  \ldots  \,\leq \,a_s  \,\leq \, 1$  with $s \in \mathbb{N}$,  then
  		\[
  	      	\displaystyle  \sum_{l=1}^{s-1} \, \vert  f_0(a_l,t) \,-\,f_0(a_{l+1},t) \vert     \,\,<\,\,C,
  		\]
  		and
  			\[
  			\displaystyle  \sum_{l=1}^{s-1} \, \vert  f_0(t,a_l) \,-\,f_0(t,a_{l+1}) \vert     \,\,<\,\,C,
  			\]
  		%There exists  a constant $c >0$ such that
  		%\[
  		%   c\,\leq \, \vert  g(x) \,-\,g(y) \vert   \,\,\,\,\,\forall x \in  A_l, \,\,y \in A_{l^{\prime}}\,\,\,\,\,\text{with}  \,\,l \neq l^{\prime}, \, l,l^{\prime} \in [K]. 
  		% \]
  	\end{enumerate}
  	for  all $t \in [0,1]$.
  %	We  write  $ \ell(i)  \,:=\,  l $  if  $\xi_i \in A_l$.
  	
  \end{assumption}

 Importantly,  we allow for  a flexibility of the graphon by only requiring   that has bounded variation. The most restrictive assumption is perhaps that   $g$ is piecewise monotonic. As  for  the   condition  expressed by (\ref{eqn:bilip}),  we acknowledge that this requirement  appeared  in the analysis of \cite{chan2014consistent}. 
  
  \begin{theorem}
  	\label{thm:bv_class}
  	Suppose  that Assumption \ref{as5} holds.  Let $\hat{\tau}$ be constructed  as in Section \ref{sec:class_of_estimators}  by setting $d \,:=\, \hat{d}_1$ (see Equation \ref{eqn:distance_sorting}). Then  for an appropriate choice of $\lambda$,  the corresponding estimator  $\hat{\theta}$ in  (\ref{eqn:estimator})  satisfies  that  
  	\[
  	\displaystyle    \frac{1}{n^2}   \sum_{ i,j \in [n-m],  \,i < j  }  \left(   \hat{\theta}_{ \hat{\tau}(i), \hat{\tau}(j)   }  \,-\,   \theta^*_{\hat{\tau}(i), \hat{\tau}(j)  }  \right)^2   \,\,=\,\,  
  	O_{\mathbb{P}} \left(   \frac{  \, \log^{  \frac{3}{2} } n  }{\sqrt{n} \,}   \right).
  	% \vert  \theta_{ \hat{\tau}(i),j}^* \,-\, \theta_{\hat{\tau}(i+1),j}^*\vert    \,\,=\,\, O_{\mathbb{P}}\left(     \frac{  \log^{\frac{(1+\alpha)}{2}   }n  }{n^{  \frac{\alpha}{2} }}  \right).
  	\]
  \end{theorem}

  Thus, on the class of functions  from Assumption \ref{as5},  the SAS-FL  estimator  attains the rate  $n^{-1/2}$,  which matches  the theoretical  result  from \cite{zhang2015estimating}  but for the class of piecewise Lipschitz functions.
  
  Interestingly,   we are the first  to study  graphon estimation with the  bounded  variation assumption.

 \section{Experiments}
 \label{sec:experiments}

The purpose  of this section is to shed some lights on the empirical performance of the class of estimators proposed in this paper. Evaluations of performance  are presented  next on both simulated and real networks.

\subsection{Network denoising}

We begin by  considering   examples   of   simulated  data   that are intended  to  test  the validity  of our general  class of methods on qualitatively different scenarios. The specifications  of $\hat{d}$ that we consider  are those described in Section \ref{sec:d}.

As benchmarks  we consider  the   following approaches.  The neighborhood  smoothing  method (NS)   from \cite{zhang2015estimating},   universal singular value thresholding (USVT) algorithm from \cite{chatterjee2015matrix}, and the sort and smooth (SAS)  method from   \cite{chan2014consistent}.

In all comparisons, the MSE  is used as a measure of performance.  Four different  scenarios are constructed. In the first scenario  the ground truth is  the stochastic block model with 12 communities. In our second example,  $f_0$ is taken as  piecewise smooth, where  locally  the function behaves like linear combinations of the  $\sqrt{\cdot}$  function applied to each coordinate.
We also consider a piecewise  constant   model (not a stochastic block). In the latter, the  degree function behaves  locally as a constant  making estimation difficult for both SAS  and SAS-FL. Our final example consists of $f_0$  being a polynomial of two variables.  Figures \ref{fig:ex1}-\ref{fig:ex4}  offer a visualization of the examples.
%For a visualization of the examples, we refer the reader to Figures \ref{fig:ex1}-\ref{fig:ex4}.

\begin{figure}[bp!]%[t!]
	\begin{center}
		\includegraphics[width=1.55in,height= 1.59in]{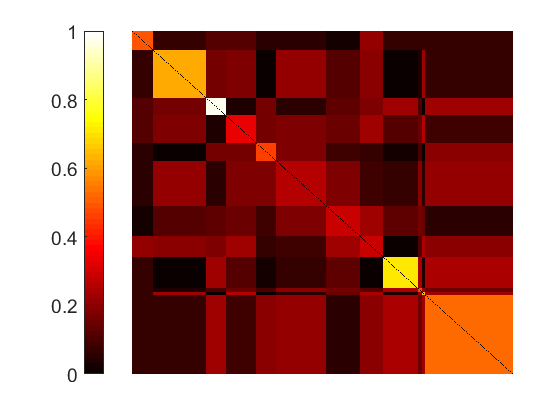}
		%[width=3.0in,height= 2.59in]{ex1_P_K12.png}
		\includegraphics[width=1.55in,height= 1.59in]{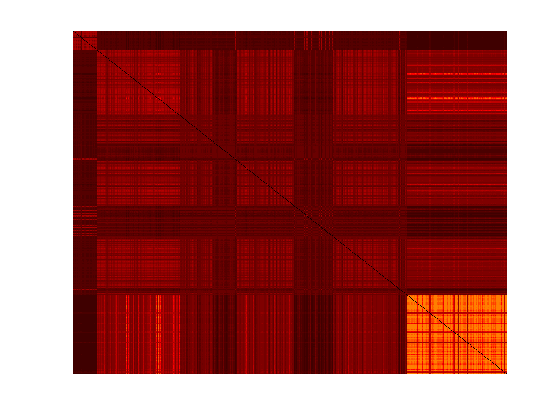}
		\includegraphics[width=1.55in,height= 1.59in]{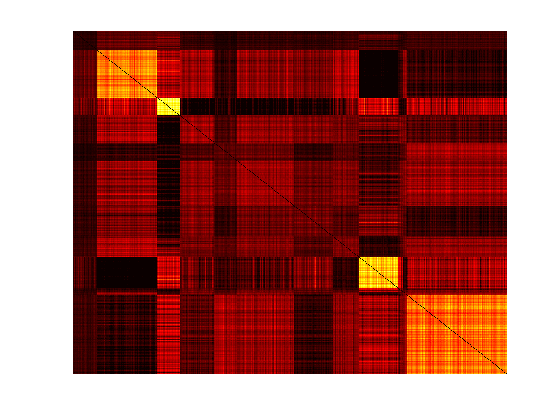}
		\includegraphics[width=1.55in,height= 1.59in]{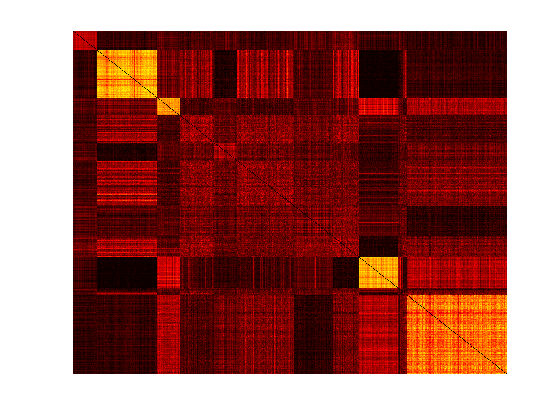}\\
		\includegraphics[width=1.55in,height= 1.59in]{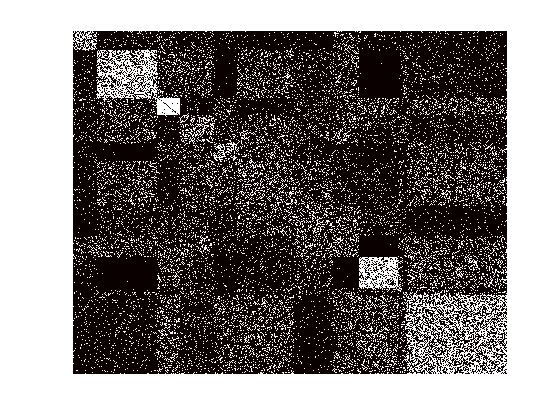}
		\includegraphics[width=1.55in,height= 1.59in]{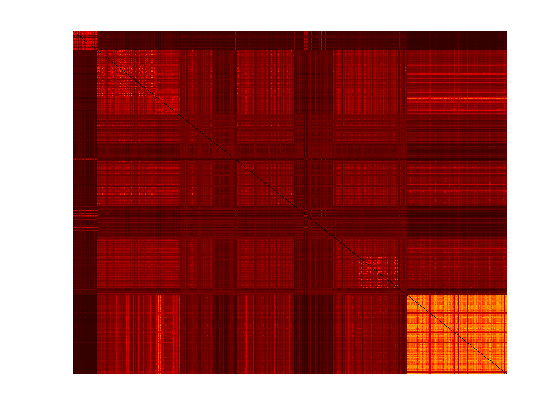}
		\includegraphics[width=1.55in,height= 1.59in]{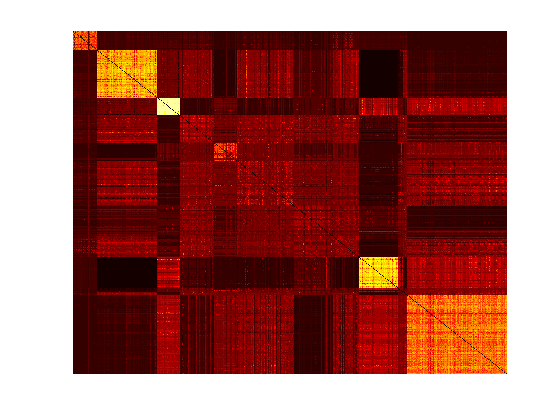}
		\includegraphics[width=1.55in,height= 1.59in]{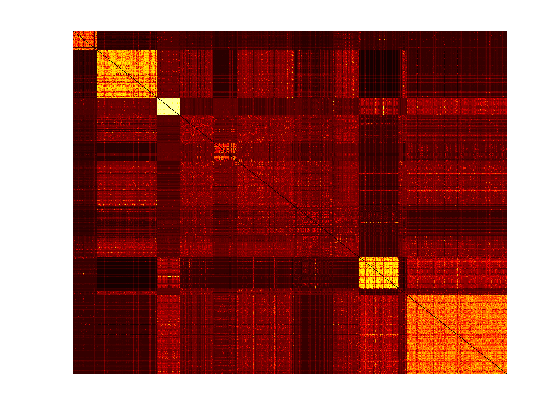}\\
		%		ch column of plots corresponds to an example, while each row  to a method. The first row depict
		\caption{\label{fig:ex1} 
			The top left in the first  row  shows a realization of the matrix of probabilities $P$ for Example 1, here $n= 500$. Then  from left to right the panels in the first row correspond to the methods  SAS, USVT, and NS.  In the second row the leftmost plot  corresponds to a realization of the incidence matrix (A) drawn  with  the parameters in  $P$ from the first row. Then from left to right the remaining plots are associated to  SAS-FL, NN-FL, and $\ell_1$-FL.
					}
		%Eas the true generating function. In all cases the plot  refer to either the true function or an estimate evaluated at an evenly spaced  grid of $100\times100$ in $[0,1]\times[0,1]$. For the estimates, the estimation was performed using $6500$ i.i.d samples}% on the covariate space and their evaluations on the true function plus  standard Gaussian noise. }
	\end{center}
\end{figure}

\begin{figure}[bp!]%[t!]
	\begin{center}
		\includegraphics[width=1.55in,height= 1.59in]{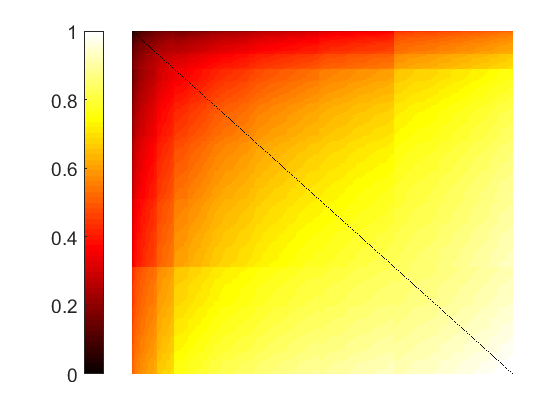}
		%[width=3.0in,height= 2.59in]{ex1_P_K12.png}
		\includegraphics[width=1.55in,height= 1.59in]{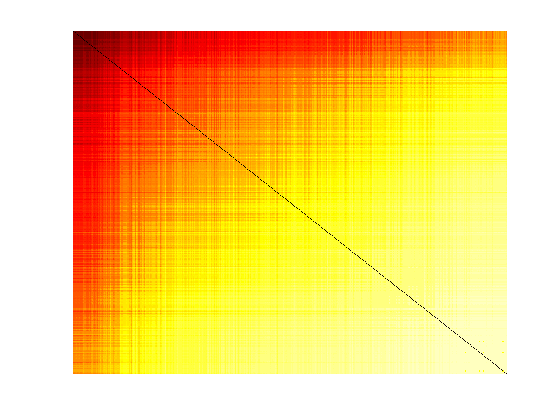}
		\includegraphics[width=1.55in,height= 1.59in]{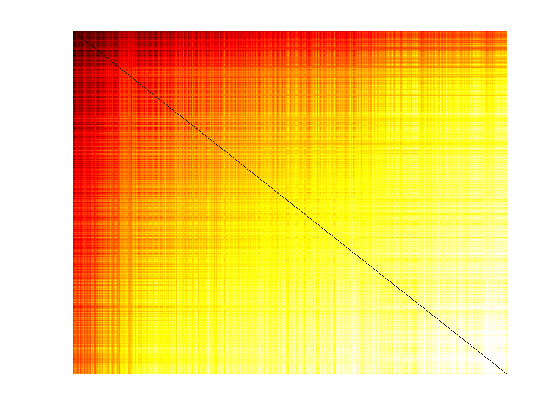}
		\includegraphics[width=1.55in,height= 1.59in]{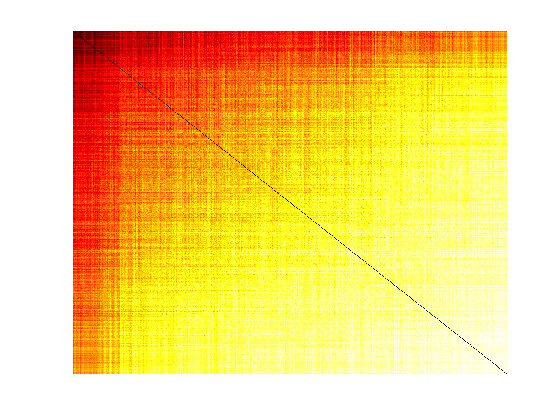}\\
		\includegraphics[width=1.55in,height= 1.59in]{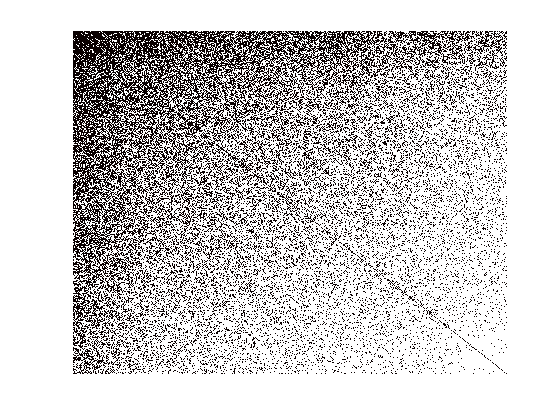}
		\includegraphics[width=1.55in,height= 1.59in]{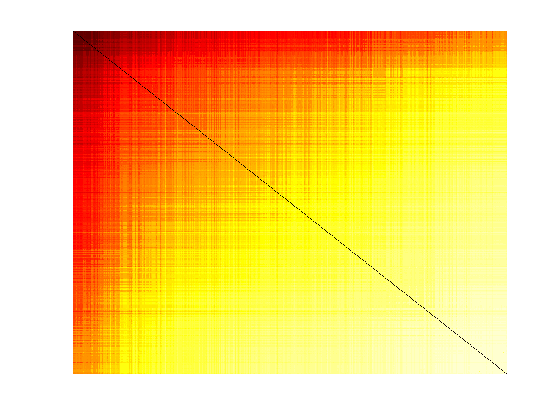}
		\includegraphics[width=1.55in,height= 1.59in]{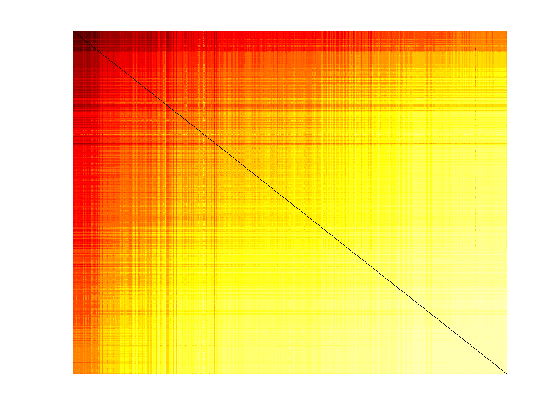}
		\includegraphics[width=1.55in,height= 1.59in]{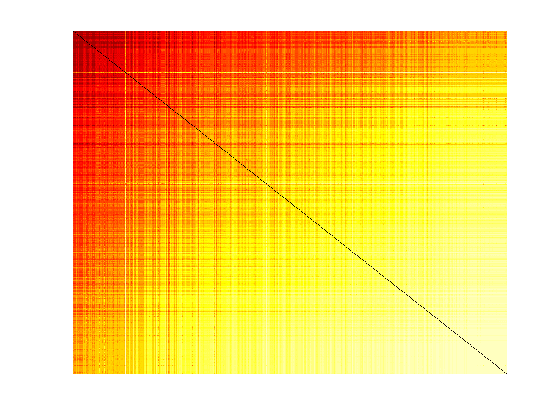}\\
		%		ch column of plots corresponds to an example, while each row  to a method. The first row depict
		\caption{\label{fig:ex2} 
			The top left in the first  row  shows a realization of the matrix of probabilities $P$ for Example 2, here $n= 500$. Then  from left to right the panels in the first row correspond to the methods  SAS, USVT, and NS.  In the second row the leftmost plot  corresponds to a realization of the incidence matrix (A) drawn  with  the parameters in  $P$ from the first row. Then from left to right the remaining plots are associated to  SAS-FL, NN-FL, and $\ell_1$-FL.
			%panel shows the regression function evaluated at an evenly spaced  grid of size $100\times100$ in $[0,1]^2$. The top right panel shows  the estimated regression function obtained  with the $K$-NN-FL  setting $K =  5$. The bottom two panels from left to right show the corresponding estimated functions  using CART  and $K$-NN ($K=12$) respectively.
		}
		%Eas the true generating function. In all cases the plot  refer to either the true function or an estimate evaluated at an evenly spaced  grid of $100\times100$ in $[0,1]\times[0,1]$. For the estimates, the estimation was performed using $6500$ i.i.d samples}% on the covariate space and their evaluations on the true function plus  standard Gaussian noise. }
	\end{center}
\end{figure}
%\newpage

\begin{figure}[bp!]%[t!]
	\begin{center}
		\includegraphics[width=1.55in,height= 1.59in]{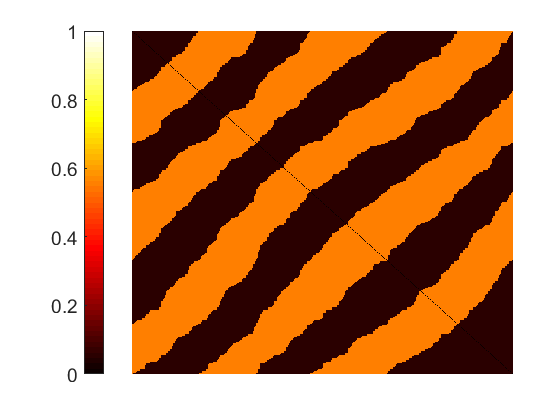}
		%[width=3.0in,height= 2.59in]{ex1_P_K12.png}
		\includegraphics[width=1.55in,height= 1.59in]{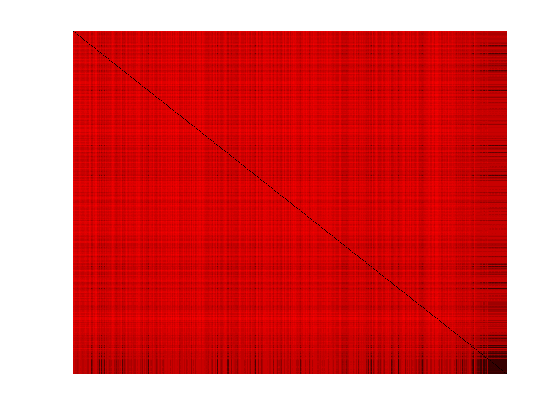}
		\includegraphics[width=1.55in,height= 1.59in]{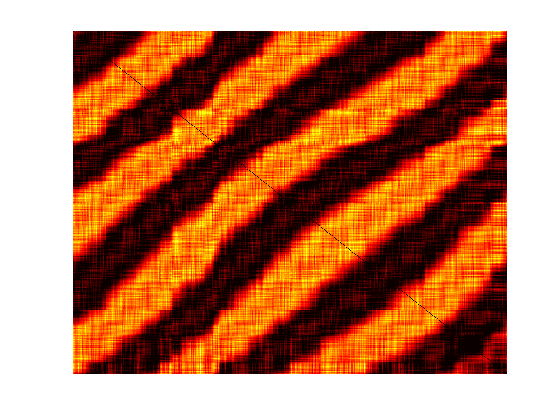}
		\includegraphics[width=1.55in,height= 1.59in]{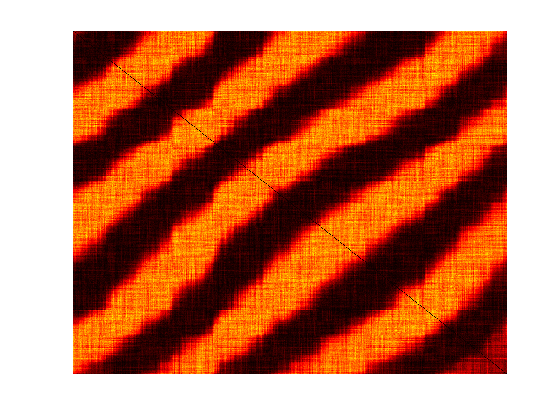}\\
		\includegraphics[width=1.55in,height= 1.59in]{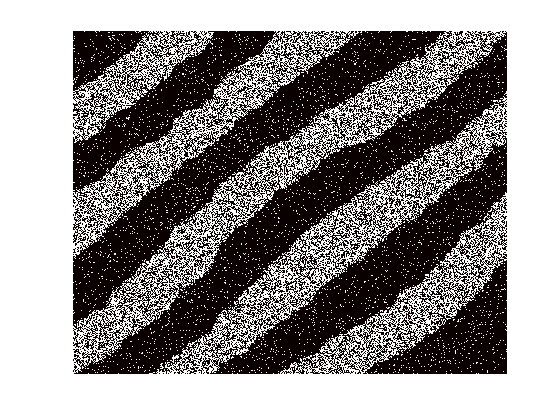}
		\includegraphics[width=1.55in,height= 1.59in]{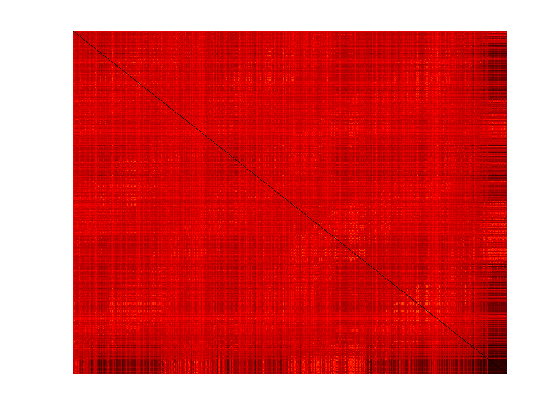}
		\includegraphics[width=1.55in,height= 1.59in]{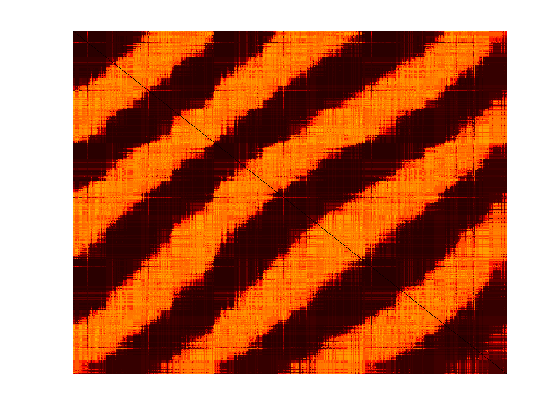}
		\includegraphics[width=1.55in,height= 1.59in]{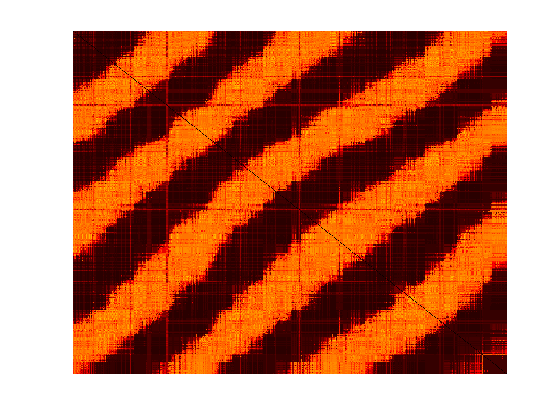}\\
		%		ch column of plots corresponds to an example, while each row  to a method. The first row depict
		\caption{\label{fig:ex3} 
				The top left in the first  row  shows a realization of the matrix of probabilities $P$ for Example 3, here $n= 500$. Then  from left to right the panels in the first row correspond to the methods  SAS, USVT, and NS.  In the second row the leftmost plot  corresponds to a realization of the incidence matrix (A) drawn  with  the parameters in  $P$ from the first row. Then from left to right the remaining plots are associated to  SAS-FL, NN-FL, and $\ell_1$-FL.
		}
		%Eas the true generating function. In all cases the plot  refer to either the true function or an estimate evaluated at an evenly spaced  grid of $100\times100$ in $[0,1]\times[0,1]$. For the estimates, the estimation was performed using $6500$ i.i.d samples}% on the covariate space and their evaluations on the true function plus  standard Gaussian noise. }
	\end{center}
\end{figure}

\begin{figure}[bp!]%[t!]
	\begin{center}
		\includegraphics[width=1.55in,height= 1.59in]{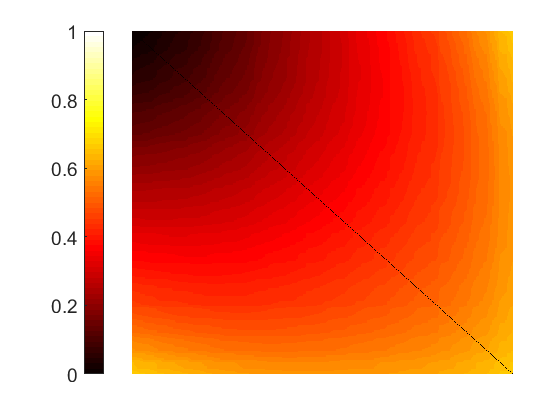}
		%[width=3.0in,height= 2.59in]{ex1_P_K12.png}
		\includegraphics[width=1.55in,height= 1.59in]{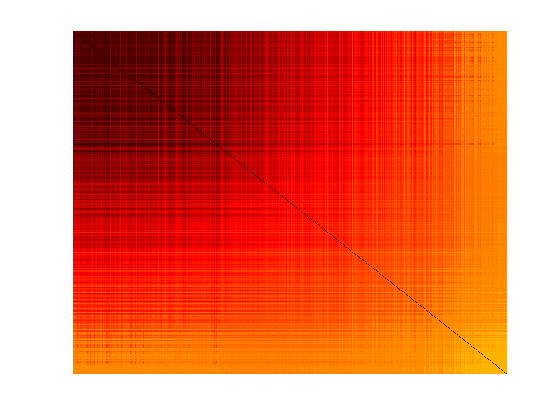}
		\includegraphics[width=1.55in,height= 1.59in]{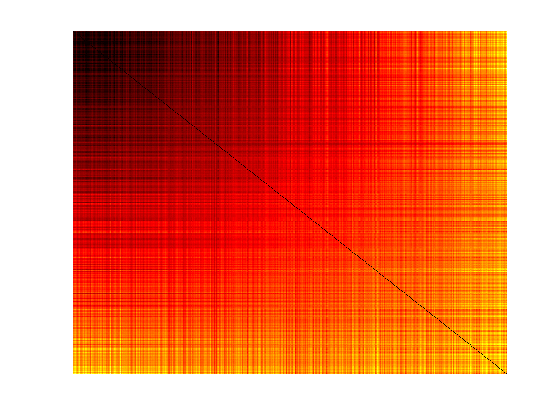}
		\includegraphics[width=1.55in,height= 1.59in]{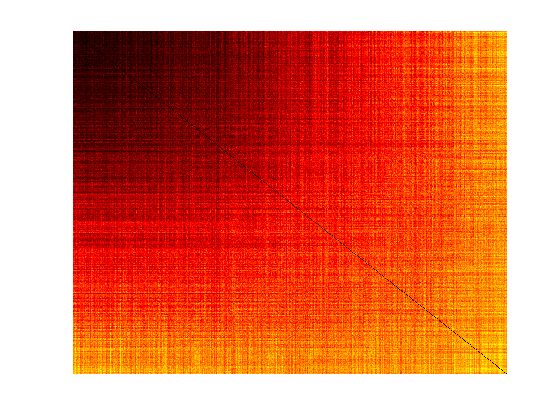}\\
		\includegraphics[width=1.55in,height= 1.59in]{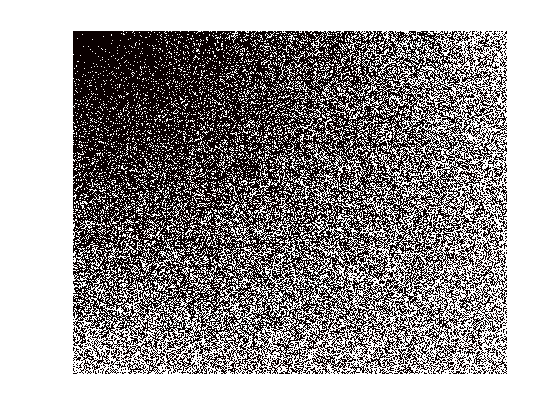}
		\includegraphics[width=1.55in,height= 1.59in]{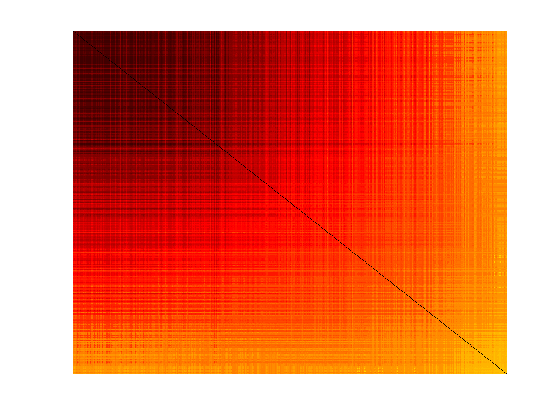}
		\includegraphics[width=1.55in,height= 1.59in]{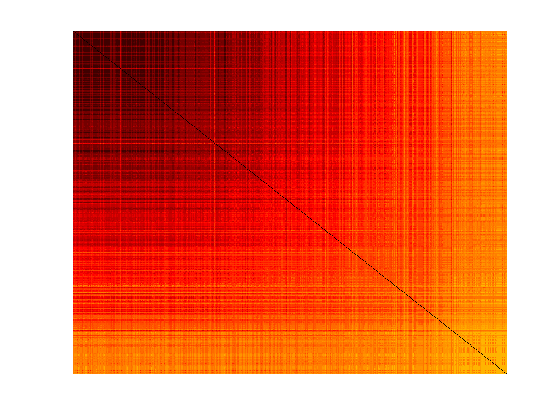}
		\includegraphics[width=1.55in,height= 1.59in]{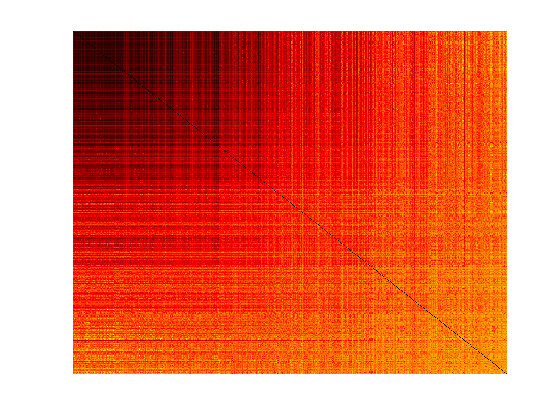}\\
		%		ch column of plots corresponds to an example, while each row  to a method. The first row depict
		\caption{\label{fig:ex4} 
					The top left in the first  row  shows a realization of the matrix of probabilities $P$ for Example 4, here $n= 500$. Then  from left to right the panels in the first row correspond to the methods  SAS, USVT, and NS.  In the second row the leftmost plot  corresponds to a realization of the incidence matrix (A) drawn  with  the parameters in  $P$ from the first row. Then from left to right the remaining plots are associated to  SAS-FL, NN-FL, and $\ell_1$-FL.
		}
		%Eas the true generating function. In all cases the plot  refer to either the true function or an estimate evaluated at an evenly spaced  grid of $100\times100$ in $[0,1]\times[0,1]$. For the estimates, the estimation was performed using $6500$ i.i.d samples}% on the covariate space and their evaluations on the true function plus  standard Gaussian noise. }
	\end{center}
\end{figure}

\begin{table}[t!]
	\centering
	\caption{	\label{tab:sim}Simulation results for Examples 1, 2, 3 and 4, see Figures \ref{fig:ex1}--\ref{fig:ex4} in that order.  Comparisons between the true and estimated probability matrices for different methods  given samples from each example. The acronyms here are explained the text. The Mean squared error (MSE)  is multiplied by a constant. }
	\begin{subtable}{1\textwidth}
		\centering
		\centering
		\caption{\label{tab:sim1}Mean squared error,  times 1000, averaging over 50 Monte Carlo simulations, for different methods  given samples from  Example 1. }
		\smallskip
		\begin{small}
			\begin{tabular}{p{3.5pc} p{3.5pc} p{3.5pc} p{3.5pc} p{3.5pc} p{3.5pc} p{3.5pc} }
				n           & NN-FL         & L1-FL  & SAS-FL    & SAS     & USVT & NS \\
				500         &\textbf{2.8}   & 4.9    & 13.5    & 15.2    & 5.1  & 3.9  \\			
				1000        &\textbf{1.5}   & 2.7    & 10.5    & 13.1    & 2.1  & 2.5 \\
				2000        &\textbf{1.2}   & 1.6    &9.2      &12.3     &1.6   &1.9\\
				%	50000       & 63.27 & \textbf{63.21} & 63.55 & 65.20   & 69.85  & 63.23    & 64.15  & 63.44     & 63.26\\
				%	100000      & 63.27 & 63.23 & 63.59 & 63.79   & 69.89  & 63.21    & 63.86  & 63.39     & \textbf{63.18}\\
			\end{tabular}
		\end{small}
	\end{subtable}
	\begin{subtable}{1\textwidth}
		\centering
		\centering
		\caption{\label{tab:sim2} Mean squared error,  times 1000, averaging over 50 Monte Carlo simulations, for different methods  given samples from  Example 2. }
		\medskip
		\begin{small}
			\begin{tabular}{p{3.5pc} p{3.5pc} p{3.5pc} p{3.5pc} p{3.5pc} p{3.5pc} p{3.5pc} }
				n           & NN-FL  & L1-FL  & SAS-FL              & SAS     & USVT & NS \\
				500         &1.7     & 4.2    & \textbf{0.9}      & 1.0     &1.9   &   2.9 \\			
				1000        &1.0     & 2.4    &\textbf{0.47}      &0.53     &0.82  &  1.7  \\	
				2000        &0.88    &1.7     &\textbf{0.33}      &0.35     &0.55  &1.3    \\	
				%	50000       & 63.27 & \textbf{63.21} & 63.55 & 65.20   & 69.85  & 63.23    & 64.15  & 63.44     & 63.26\\
				%	100000      & 63.27 & 63.23 & 63.59 & 63.79   & 69.89  & 63.21    & 63.86  & 63.39     & \textbf{63.18}\\
			\end{tabular}
		\end{small}
	\end{subtable}
	\begin{subtable}{1\textwidth}
		\centering
		\centering
		\caption{\label{tab:sim3} Mean squared error,  times 1000, averaging over 50 Monte Carlo simulations, for different methods  given samples from  Example 3. }
		\medskip
		\begin{small}
			\begin{tabular}{p{3.5pc} p{3.5pc} p{3.5pc} p{3.5pc} p{3.5pc} p{3.5pc} p{3.5pc} }
				n           & NN-FL       & L1-FL   & SAS-FL       & SAS     & USVT & NS \\
				500         &\textbf{8.2} &9.9      &56.5        &60.7     &9.4   &8.7    \\				
				1000        &\textbf{5.7} &6.3      &47.9        &59.7     &6.4   &6.4    \\	
				2000        &5.1          &\textbf{5.0}&44.6     &59.4     &5.3  &5.3    \\	
				%	50000       & 63.27 & \textbf{63.21} & 63.55 & 65.20   & 69.85  & 63.23    & 64.15  & 63.44     & 63.26\\
				%	100000      & 63.27 & 63.23 & 63.59 & 63.79   & 69.89  & 63.21    & 63.86  & 63.39     & \textbf{63.18}\\
			\end{tabular}
		\end{small}
	\end{subtable}
	\begin{subtable}{1\textwidth}
		\centering
		\centering
		\caption{\label{tab:sim4}Mean squared error,  times 1000, averaging over 50 Monte Carlo simulations, for different methods  given samples from  Example 4. }
		\medskip
		\begin{small}
			\begin{tabular}{p{3.5pc} p{3.5pc} p{3.5pc} p{3.5pc} p{3.5pc} p{3.5pc} p{3.5pc} }
				n           & NN-FL  & L1-FL  & SAS-FL             & SAS     & USVT & NS \\
				500         &2.0     &4.3     &\textbf{1.3 }     &1.4      &1.8  &3.6    \\				
				1000        &1.3     &2.5     &\textbf{0.67}     &0.71     &0.88 &2.2    \\	
				2000        &1.1     &1.9     &\textbf{0.48}     &0.50    &0.59  &1.7    \\	
				%	50000       & 63.27 & \textbf{63.21} & 63.55 & 65.20   & 69.85  & 63.23    & 64.15  & 63.44     & 63.26\\
				%	100000      & 63.27 & 63.23 & 63.59 & 63.79   & 69.89  & 63.21    & 63.86  & 63.39     & \textbf{63.18}\\
			\end{tabular}
		\end{small}
	\end{subtable}
\end{table}

For the methods based on fused lasso denoising, namely SAS-FL, NN-FL, and $\ell_1$-FL   we choose the tunning  parameter  $\lambda$  by cross-validation. This is done by selecting  the best value of $\lambda$ out of 30  candidates, by erasing $20\%$  of the data points and replacing them by zeros, and then performing predictions based on the remaining data.

The results for each scenario are given in Table \ref{tab:sim}. These are obtained by averaging over 50  Monte Carlo simulations, and  for  values of  $n \in \{500,1000,1500\}$. In Table \ref{tab:sim}, we see that in most cases  the best method is either NN-FL  or  SAS-FL. Even  for the stochastic block model example, the method NN-FL seems to have the best performance.

Moreover,  in Example 3, we see that SAS and SAS-FL  suffer  greatly  due to the   nearly constant behavior of the degree function ($g_1$, $g_2$ in Assumption \ref{as5}). In  contrast,  NN-FL and $\ell_1$-FL  are not affected by the degree issue and offer strong performance.

Figures  \ref{fig:ex1}-\ref{fig:ex4} allow us to   visualize the  comparisons of different methods in the examples considered. In Figure \ref{fig:ex1}, we can see that NN-FL gives a more  detailed recovery of the blocks compared to  all the other methods. As for Figure \ref{fig:ex2},  we see that  with $n = 500$,  all the competing methods are comparable, although, based on MSE,  SAS-FL is the best  approach. In Figure \ref{fig:ex3},  we clearly see the effect of the degree on the performance of SAS and SAS-FL. %Both methods are greatly affected under Example  3, while the  other methods perform  reasonable, with NN-FL  being the best.
% Figure \ref{fig:ex2}  shows  clearly  how  SAS-FL and SAS  are affected by the properties of the degree function, while the other methods are not affected. We also note that Figure \ref{fig:ex3}  illustrates

\subsection{Link prediction}

We now validate the methods studied in this paper in the task of link prediction.  To  that end, we consider  two different  datasets. For our first example we use the Protein230 dataset  from \cite{butland2005interaction}.  This consists of 230 proteins  and their interactions encoded in 595 edges.  In our second  example, we use the Yeast protein interaction network from \cite{butland2005interaction}.  This is a larger  network  consisting of 2361 nodes and 6646  edges.  %We refer to this dataset as Example 2.

Using the data described  above, we evaluate the prediction  performance  of different methods as follows. In each case,  we remove some  observations of the matrix $A \in \mathbb{R}^{n \times n}$, thus rather  than observing $A$, we assume that the data is $\tilde{A}$ where 
\[
   \begin{array}{lll}
   \tilde{A}_{i,j} & =&  \rho_{i,j}\,A_{i,j},\,\,\,\,\,\\ 
    \rho_{i,j}& \sim^{i.i.d} &\text{Bernoulli}(0.8),  \,\,\,\,\,\,\forall  i,j \in [n].  
   \end{array}
   \]
For each data example  we generate  $50$  trials  of  $\tilde{A}$, and for  each instance  of $\tilde{A}$ we  fit different estimators.  For each estimator  we compute  the area under the curve  (AUC) of  the receiver
operating characteristic (ROC). We then  report the average  AUC-ROC  and refer to it simply as AUC-ROC.

With the setting described above,  
 Table \ref{tab:real}  reports  the average  AUC-ROC  associated to the competing methods in each of the considered examples. We can see that for both  examples,  NN-FL  and $\ell_1$-FL are the most competitive estimators. As a sanity check, we also computed the  area under the  precision recall curve, and found that in both cases the best approach  was NN-FL.

\begin{table}[t!]
	\centering
	\caption{	\label{tab:real} Average AUC-ROC for the competing methods under Examples 1 and 2.  }
	%	\begin{subtable}{1\textwidth}
	\centering
	%\centering
	%\caption{\label{tab:real}Mean squared error,  times 1000, averaging over 50 Monte Carlo simulations, for different methods  given samples from  Example 1. }
	\smallskip
	\begin{small}
		\begin{tabular}{p{3.5pc} p{3.5pc} p{3.5pc} p{3.5pc} p{3.5pc} p{3.5pc} p{3.5pc} }
			Example           & NN-FL         & $\ell_1$-FL  & SAS-FL    & SAS     & USVT & NS \\
			1                 &\textbf{0.84}  & 0.83   &0.76      &0.84     &0.60      &0.69  \\			
			2                 & 0.92          &\textbf{0.93} & 0.80       &0.82         &  0.39           & 0.92 \\	% \\
			%2000        &\textbf{1.2}   & 1.6    &9.2      &12.3     &1.6   &1.9\\
			%	50000       & 63.27 & \textbf{63.21} & 63.55 & 65.20   & 69.85  & 63.23    & 64.15  & 63.44     & 63.26\\
			%	100000      & 63.27 & 63.23 & 63.59 & 63.79   & 69.89  & 63.21    & 63.86  & 63.39     & \textbf{63.18}\\
		\end{tabular}
	\end{small}
	%	\end{subtable}
\end{table}

%\newpage

\section{Conclusion}

We have  studied a  novel class of graphon  estimators based on a two--step approach that  combines the nearest neighbor algorithm  with  the  graph fused lasso. 
Overall  the  estimators  seem to perform  reasonably  in both simulated and real data. 

Statistical  guarantees have been provided, although some questions  remain open. For instance, we have not studied  the statistical performance  when the graphon is piecewise H\"{o}lder with exponent in the interval $[1/2,1]$. 

We also  leave for future work to understand  the convergence properties  $\ell_1$-FL  which also seems like a reasonable approach, at least in the examples considered here.

\appendix

\section{Proofs}

We  assume that  $f_0(x,y) =  f_0(y,x)$ for all $x,y \in [0,1]$. This is not  important in the proofs  but it simplifies the notation.

\subsection{Important lemmas}

First we recall   Theorem 2 from \cite{hoeffding1963probability}.

\begin{theorem}
	\label{thm:Hoeffding}
	(\textbf{Hoeffding})  Let $Z_1,\ldots,Z_n$ be centered independent  random variables  satisfying $a \leq Z_i \leq b$  for  all $i= 1,\ldots,n$. Then,
	for all $\delta \,\,>\,\, 0$,   we have, 	with probability at least $1 \,-\, \delta$,  that
	\[
	\left\vert  \frac{1}{n}  \sum_{i=1 }^n\,Z_ n  \right\vert  \,\,\leq\,\,\sqrt{ \frac{c\,\log\left(2/\delta\right) }{2\,n}   } 
	\]
	where $c  \,\,=\,\,n^{-1}\sum_{i=1}^{n}(a_i \,-b_i)^2 $.
\end{theorem}

\begin{lemma}
	\label{lem:permutation}
	Let $\epsilon \in \mathbb{R}^d$  a  random vector and let $\hat{P} \in  \mathcal{S}_d$  be random permutation such that $\epsilon \indep \hat{P}$. If the coordinates  of $\epsilon $ are independent with mean zero and belong to $[-1,1]$, then
	%$\epsilon$ follows the same distribution as $\tilde{\epsilon}  \,\,=\,\,   (\epsilon_{ \hat{P}(1)},\ldots,\epsilon_{\hat{P}(d) })$.
	$$\mathbb{E}\left( e^{ \sum_{i=1}^{d  } s_i\,\epsilon_{ \hat{P}(i)}}  \right)  \,\,\leq \,\, e^{ \|s\|_2^2/2 },\,\,\,\,\,\forall s \in \mathbb{R}^d.  $$
\end{lemma}

\begin{proof}
	Simply note that
	\[
	%	\begin{array}{lll}
	\mathbb{E}\left( e^{ \sum_{i=1}^{d  } s_i\,\epsilon_{ \hat{P}(i)}}  \right)     \,\,  = 
	\,\,  \mathbb{E}\left( \mathbb{E}\left( e^{ \sum_{i=1}^{d  } s_i\,\epsilon_{ \hat{P}(i)}}  \big| \hat{P} \right) \right)\\  
	\,\,  = \,\,  \mathbb{E}\left(  \prod_{i=1}^{d} \mathbb{E}\left( e^{  s_i\,\epsilon_{ \hat{P}(i)}}  \big| \hat{P} \right) \right)\\
	%   & \leq &   \\
	\,\,  \leq \,\, \mathbb{E}\left(  \prod_{i=1}^{d} e^{s_i^2/2 } \right)\\
	\,\,=  \,\,  e^{ \|s\|_2^2/2 },
	%\end{array}
	\]
	where the first equality follows from the tower property of expectations,   the second equality holds by the independence assumption, and the third inequality is met  by Lemma 2.2 from  \cite{boucheron2013concentration}.
\end{proof}

\begin{lemma}
	\label{lem:tv:rates}
	(\textbf{Minor modification to Corollary 5 from  \cite{hutter2016optimal} }) Consider the model
	\[
	y_i \,\,=\,\,   \mu^*_i  \,\,+\,\,\varepsilon_{ i },\,\,\,\,\,\,i =1,\ldots,   m  \,=\, N_1 \times N_2,
	\]
	with  $N_l / N  =   z_l$ for some $z_l \in \mathbb{N}$,  $l = 1,2$. Here $N$  is a quantity that can grow to infinity.  Moreover, we assume that  for each  $j \in [m]$  we have that  $\varepsilon_j   \,=\, \epsilon_i $  for some $i \in [\tilde{m}]$,   where  $\epsilon \in \mathbb{R}^{\tilde{m}}$ with $\tilde{m}  \asymp m$. In addition, we assume that  $s_i \,=\, \vert \{ j \in [m] \,:\, \varepsilon_j   \,=\, \epsilon_i  \}\vert  \leq  \kappa $  for all $i \in [\tilde{m}]$,  for some positive constant $\kappa$. We also assume that  the coordinates of  $\epsilon   \,\,=\,\,   (\epsilon_1,\ldots,\epsilon_{\tilde{m}})$  are independent and satisfy $\mathbb{E}(\epsilon_i) \,=\, 0$,  $\epsilon_i \in [-1,1]$. 
	Let  
	\[
	\hat{\mu }  \,\,=\,\,  \underset{\mu \in \mathbb{R}^{m}  }{\arg \min} \,\,\frac{1}{2}\sum_{i=1}^{m}(y_{\hat{P}}(i)  \,-\, \mu_i )^2    \,\,+\,\,   \lambda\,m\,\|D\mu\|_1,
	\]
	where $D$ is the incidence matrix of  $N_1 \times N_2$ grid graph, $\lambda >0$ is a tunning parameter, and  $\hat{P} \in \mathcal{S}_{\tilde{m}}$ is  a random permutation independent of $\epsilon$. Then there exists a constant $c>0$ such  that for any $\delta >0$,    if   $\lambda$ is chosen as 
	\[
	\lambda  \,\,=\,\, \frac{c\,\log N  \,\sqrt{  \log\left(   \frac{e\,m}{\delta}    \right)    } }{m}, 
	\] 
	then 
	\[
	\mathbb{P}\left( \frac{1}{m} \|  \hat{\mu} \,-\,\mu^* \|_2^2 \,\,\geq \,\, C\frac{\left[\log m\,\sqrt{ 2\log( e\,m/\delta  )  } \|D\mu^* \|_1\, \,+\,  \log^2 m\,\log(e/\delta)       \right]   }{m}\right)      \,\,< \,\,\frac{2\,z_1\,z_2}{\delta},
	\]
	for some constant $C > 0$.
\end{lemma}

\begin{proof}
	We proceed as in the basic inequality argument from the proof of  Theorem 3 in \cite{wang2016trend}. Letting $\tilde{\varepsilon}  \,=\, (\varepsilon_{\hat{P}(1)  }, \ldots,  \varepsilon_{\hat{P}(m)  } )$, we obtain
	\begin{equation}
	\label{eqn:basic_inequality}
	\frac{1}{2}\, \|   \hat{\mu} \,-\,\mu^*  \|_2^2  \,\,\leq \,\,  \tilde{\varepsilon}^T(  \hat{\mu} \,-\, \mu^* )\,\,+\,\,  m\,\lambda\left[ \|D\mu^* \|_1 \,-\, \|D\hat{\mu} \|_1  \right].  
	\end{equation}
	Note that there exists two disjoint subsets $A_1,  A_2,  \ldots,  A_{z_1\cdot z_2}$  of $\{1,\ldots,m\}$ such that in the graph with incidence matrix $D$, the elements of $A_1$ lie on sub--grid graph of dimensions $N \times N$, the elements of $A_2$ lie on a different sub--grid of dimensions  $N \times N$ and so on. Throughout, for a vector $\mu \in \mathbb{R}^{N_1 \times N_2}$, we write 
	$\mu^j   \,\,=\,\,(\mu_{i_1^j},\ldots, \mu_{i_{N^2}^j}) $   where $i_1^j < \ldots < i_{N^2}^j $, with $A_j  \,=\,\{i_1^j,\ldots,i_{N^2}^j   \}$ for $j = 1,2,\ldots,z_{1}\cdot z_2$. Moreover,  we denote by $D^1$  the incidence matrix of grid graph with dimensions $N \times N$.
	
	Using the notation above, and writing $\Pi$ as the projection onto the span of $\boldsymbol{1}  \in \mathbb{R}^{N^2}$, we obtain  that there exists  permutation matrices $P_1,\ldots,P_{z_1\cdot z_2}  \in  \mathbb{R}^{ N^2 \times N^2 } $  such that 
	\begin{equation}
	\label{eqn:emp_part}
	\begin{array}{lll}
	\tilde{\varepsilon}^T(\hat{\mu} \,-\, \mu^*)    & = & \displaystyle \sum_{j=1}^{z_1\cdot z_2 } \,(\tilde{\varepsilon}^j)^T(\hat{\mu}^{j} \,-\, (\mu^*)^j) \\
	% (\tilde{\epsilon}^1)^T(\hat{\mu}^{1} \,-\, (\mu^*)^1) \,\,+\,\, (\tilde{\epsilon}^1)^T(\hat{\mu}^{2} \,-\, (\mu^*)^2) \\
	& = &  \displaystyle \sum_{j=1}^{z_1\cdot z_2} \left[  ( \Pi\tilde{\varepsilon}^j)^T(\hat{\mu}^{j} \,-\, (\mu^*)^j) \,\,+\,\,((I-\Pi)\tilde{\varepsilon}^j)^T(\hat{\mu}^{j} \,-\, (\mu^*)^j)   \right]\\
	& \leq &  \displaystyle \sum_{j=1}^{z_1\cdot z_2} \, \|\Pi\tilde{\varepsilon}^j \|_2\,\| \hat{\mu}^{j} \,-\, (\mu^*)^j\|_2  \,\,+\,\,  \underset{1 \leq j \leq z_{1}\cdot z_2  }{\max} \| ((D^1)^+)^T  P_j \tilde{\varepsilon}^j  \|_{ \infty }\left[ \|D\mu^* \|_1 \,+\, \|D\hat{\mu} \|_1 \right]   \\
	& \leq & \left[ \underset{1 \leq j \leq z_{1}\cdot z_2  }{\max} \, \|\Pi\tilde{\varepsilon}^j \|_2\right] z_1\,z_2\,\| \hat{\mu} \,-\, \mu^*\|_2  \,\,+\,\,  \underset{1 \leq j \leq z_{1}\cdot z_2  }{\max} \| ((D^1)^+)^T  P_j \tilde{\varepsilon}^j  \|_{ \infty }\left[ \|D\mu^* \|_1 \,+\, \|D\hat{\mu} \|_1 \right]   \\
	\end{array}
	\end{equation}	
	where we have proceeded as in the proof of Theorem 2 from \cite{hutter2016optimal}, using H\"{o}lder's inequality. Now, for
	$x \in \mathbb{R}^{N^2}$     we have that
	\[
	x^T   \,\tilde{\varepsilon}^j   \,\,=\,\, \displaystyle \sum_{l=1}^{N^2}  x_l  \tilde{\varepsilon}_{  i_l^j }  \,\,=\,\, \displaystyle \sum_{l=1}^{N^2}  x_l  \varepsilon_{  \hat{P}(i_l^j) }   \,\,=\,\, \displaystyle    \sum_{s=1}^{\tilde{m}}  \epsilon_s\, \left( \sum_{ l \,:\,\varepsilon_{  \hat{P}(  i_l^j ) } \,=\,   \epsilon_s }   x_l  \right)  \,\,=:\,\,  \tilde{x}^T \epsilon,
	\]
	where $\tilde{x} \in \mathbb{R}^{ \tilde{m} }$  and satisfies
	\[
	\| \tilde{x} \|_2  \,\,=\,\,  \sqrt{ \sum_{s=1}^{\tilde{m}}  \left( \sum_{ l\,:\,\varepsilon_{  \hat{P}(i_l^j) } \,=\,   \epsilon_s }   x_l  \right)^2   }  \,\,\leq \,\,
	\sqrt{ \kappa   \, \sum_{s=1}^{\tilde{m}}
		\,\,   \sum_{ l \,:\,\varepsilon_{  \hat{P}(i_l^j) } \,=\,   \epsilon_s }   x_l^2  }   \,\,=\,\,  \sqrt{\kappa} \, \| x \|_2.
	\]
	Therefore, combining the above observation with  Lemma \ref{lem:permutation}, and  Corollary 2.6 from \cite{boucheron2013concentration}, and (\ref{eqn:emp_part})  it follows that
	\begin{equation}
	\label{eqn:emp_part2}
	\begin{array}{lll}
	\tilde{\epsilon}^T(\hat{\mu} \,-\, \mu^*)    & \leq & 2\,\kappa^{1/2}\,\sqrt{2\log( e/\delta ) }\,\| \hat{\mu} \,-\,\mu^* \|_2 ,\,+\,\,\rho\,\kappa^{1/2}\,\sqrt{2\log( e\,N^2/\delta ) }\left[ \|D\mu^* \|_1 \,+\, \|D\hat{\mu} \|_1 \right],    \\ 
	\end{array}
	\end{equation}
	with probability at least  $1 \,-\, 2\,z_1\,z_2/\delta$,
	where as in Proposition 4 from \cite{hutter2016optimal},  there exists a constant $c_1 >0$ such that $\rho \leq c_1 \sqrt{\log m}$. 
	
	Therefore, combining (\ref{eqn:basic_inequality})  with (\ref{eqn:emp_part2}) and using the inequality  $a\,x    - x^2/4  \leq a^2  $
	for all $a$  and $x$, we arrive to the  desired conclusion.	
\end{proof}

\subsection{Proof of Theorem 2}

Let  $\sigma $ and $\tau$  be  bijections  such that 
\[
\begin{array}{l}
\xi_{\sigma(1)}  \,<\,    \ldots  \,<\,   \xi_{\sigma(n-m)},\\
g(\xi_{\tau(1)})  \,<\,    \ldots  \,<\,   g(\xi_{\tau(n-m)}),\\
\end{array}
\]
where $\sigma(i),\tau(i) \in [n]\backslash [m]$.

We  write  $ \ell(i)  \,:=\,  l $  if  $\xi_i \in A_l$.

\begin{lemma}
	\label{lem:concentration_to_g}
	Let us assume that $n \,\leq\, c_2\,m$. Then, the event 
	\[
	\Omega  \,\,=\,\, \left\{  \xi  
	\,:\,    \underset{i \in [n]\backslash [m]  }{\max} \, \left\vert   \displaystyle  g(\xi_i) \,-\,  \frac{1}{m} \sum_{j \in  [m]} A_{i,j}    \right\vert    \,\,\leq \,\, 2\,c_2^{ \frac{1}{2} }\,\left( \frac{\log n}{n}    \right)^{  \frac{1}{2} }     \right\},
	\]
	happens with probability at least $1  \,-\,  \frac{4}{n}$.
\end{lemma}

\begin{proof}
	First  we  observe  that by Theorem  \ref{thm:Hoeffding}  and union bound 
	\begin{equation}
	\label{eqn:concentration}	
	\mathbb{P}\left( \underset{i \in [n]\backslash [m] }{\max} \,  \frac{1}{m} \left\vert   \displaystyle  \sum_{j \in  [m]} A_{i,j} \,-\,   \sum_{j \in  [m]} \theta^*_{i,j}    \right\vert   \,\,>\,\,   \sqrt{c_2} \left( \frac{\log n}{n}    \right)^{  \frac{1}{2} }    \,\,\Bigg\vert\,\,  \xi \right)  \,\,\leq \,\,    \frac{2}{n}.
	\end{equation}
	Next  we define the set
	\begin{equation}
	\label{eqn:concentration2}	
	\tilde{\Omega}  \,=\,   \left\{  \xi  
	\,:\,    \underset{i \in [n]\backslash [m]  }{\max} \, \left\vert   \displaystyle  g(\xi_i) \,-\,   \frac{1}{m}\sum_{j \in  [m]} \theta^*_{i,j}    \right\vert    \,\,\leq \,\, \sqrt{c_2}\,\left( \frac{\log n}{n}    \right)^{  \frac{1}{2} }     \right\},
	%	\mathbb{P}\left( \underset{i \in [n]\backslash [m]  }{\max} \, \left\vert   \displaystyle  g(\xi_i) \,-\,   \sum_{j \in  [m]} \theta^*_{i,j}    \right\vert   \,\,>\,\, \left( \frac{\log n}{n}    \right)^{  \frac{1}{2} }    \right)  \,\,\leq \,\,    2\,\exp\left( -\log n \right).    		
	\end{equation}
	and by  Theorem  \ref{thm:Hoeffding} (and using conditional probability) we have that   $\mathbb{P}\left( \tilde{\Omega} \right)\,\geq \,  1  \,-\, 2/n$.  Now, %if $\xi  \in \Omega$  then
	\begin{equation}
	\label{eqn:concentration3}
	\def\arraystretch{2}
	\begin{array}{l}
	\mathbb{P}\left(  \underset{i \in [n]\backslash [m] }{\max} \, \left\vert   \displaystyle \frac{1}{m}  \sum_{j \in  [m]} A_{i,j} \,-\,   g(\xi_i)   \right\vert   \,\,\leq \,\, \displaystyle 2\,c_2^{\frac{1}{2}}\,\left( \frac{\log n}{n}    \right)^{  \frac{1}{2} }   \right)   \\
	\,\,\,  \geq    \displaystyle  \int_{ \tilde{\Omega} }  \, \mathbb{P}\left(  \underset{i \in [n]\backslash [m] }{\max} \, \left\vert   \displaystyle  \frac{1}{m}  \sum_{j \in  [m]} A_{i,j} \,-\,   g(\xi_i)    \right\vert   \,\,\leq \,\, 2\,c_2^{\frac{1}{2}}\,\left( \frac{\log n}{n}    \right)^{  \frac{1}{2} }   \,\,\Bigg\vert \,\, \xi  \right)\,d\xi\\ 
	\,\,\geq \,\,\displaystyle  \int_{ \tilde{\Omega} }  \, \mathbb{P}\left(  \underset{i \in [n]\backslash [m] }{\max} \, \left\vert   \displaystyle  \frac{1}{m}  \sum_{j \in  [m]} A_{i,j} \,-\,   \sum_{j \in  [m]} \theta^*_{i,j}    \right\vert   \,\,\leq \,\, \,c_2^{\frac{1}{2}}\left( \frac{\log n}{n}    \right)^{  \frac{1}{2} }   \,\,\Bigg\vert \,\, \xi  \right)\,d\xi\\ 
	\,\,\geq \,\,  \left[ 1  \,-\,   \frac{2}{n}  \right] \displaystyle  \int_{ \tilde{\Omega} }  d\xi   \\
	\,\,\geq \,\,   1  \,-\,   \frac{4}{n}. 
	\end{array}
	\end{equation}
\end{proof}

\begin{lemma}
	\label{lem:Dvoretzky}
	Let us assume that $c_1\,m\,\,\leq \,\,  n  \,\, \leq \,\,c_2\,m$. Then the event $\tilde{\Omega}$ given as
	\[
	\underset{i \in [n - m] }{ \max }  \,\left\vert   \xi_{\sigma(i)} \,\,-\,\,\frac{i}{n\,-\, m} \right\vert  \,\,\leq\,\,\frac{2\, c_2^{ \frac{1}{2}  }}{L_1}\left( \frac{\log n}{n} \right)^{ \frac{1}{2} },
	\]
	happens with probability at least $1 \,-\,  2\,\exp(-8\,c_3\,\log n/L_1^{2} \,)$, where $c_3 \,=\, c_2\, (c_1  \,-\, 1)/c_1 $.
\end{lemma}

\begin{proof}
	This follows immediately from Dvoretzky--Kiefer--Wolfowitz inequality.
	%We notice  that  for $i \in [n-m]$
	%\[
	%   \left\vert   \xi_{\sigma(i)} \,\,-\,\,\frac{i}{n} \right\vert   \,\,\leq \,\,\left\vert   \xi_{\sigma(i)} \,\,-\,\,\frac{i}{n-m} \right\vert  \,\,+\,\,\left\vert     \frac{i}{n} \,\,-\,\,\frac{i}{n-m} \right\vert  \,\,\leq \,\,
	%\]
\end{proof}

\begin{lemma}
	\label{lem:dist}
	Let  $\Omega^{\prime}$  be the event that 
	\[
	\forall\,i,j \in [n  \,-\, m],\,\,\,\,\vert i\,-\,j\vert \,\,>\,\, \frac{12}{L_1}\sqrt{c_2\,n\,\log n}  \,\,\,\,\,\,\,\,\,\text{implies}\,\,\,\,\,\,\,\,\vert \xi_{\sigma(i)}\,-\,\xi_{\sigma(j)}\vert \,\,>\,\,\frac{8}{L_1}\left( \frac{c_2\,\log n}{n} \right)^{ \frac{1}{2}  }.
	\]
	Then  $\mathbb{P}\left(  \Omega^{\prime} \right)\,\,\geq \,\,  1  \,-\,  2\,\exp(-8\,c_3/L_1^{2} \,\log n)$ with $c_3$ as in Lemma \ref{lem:Dvoretzky}.
\end{lemma}

\begin{proof}
	Le us assume that the event $\tilde{\Omega}$  from  Lemma \ref{lem:Dvoretzky}  holds.  Let  $i,j \in [n \,-\, m]$  such that  $\vert \xi_{\sigma(i)}\,-\,\xi_{\sigma(j)}\vert   \,\,<\,\, 8/L_1 \sqrt{c_2\,\log n  /n} $. Then
	\[
	\left\vert \frac{i}{n \,-\, m}\,-\, \frac{j}{n \,-\, m}  \right\vert \,\,\leq \,\,  \left\vert \frac{i}{n \,-\, m }\,-\, \xi_{ \sigma(i) } \right\vert  \,\,+\,\, \left\vert \xi_{ \sigma(i) }\,-\, \xi_{\sigma(j) }   \right\vert \,\,+\,\, \left\vert \frac{j}{n \,-\, m}\,-\, \xi_{ \sigma(j) } \right\vert   \,\,<\,\,\frac{12}{L_1}\left( \frac{c_2\,\log n}{n} \right)^{ \frac{1}{2}  },
	\]
	and the claim follows from Lemma \ref{lem:Dvoretzky}.
\end{proof}

\begin{lemma}
	\label{lem:comparison_permutations}
	The event 
	\[
	\Omega_2  \,=\,  \left\{ 	\underset{i\in [n] \backslash [m]}{\max}  \,\left\vert  \hat{\tau}^{-1}(i) \,-\, \tau^{-1}(i)    \right\vert \,   \,\,\leq \,\,  \frac{24}{L_1} \left(  c_2\,n\,\log n \right)^{ \frac{1}{2} }  \right\}
	\]
	happens with probability at least $1 \,-\, 4/n \,-\,2\,\exp(-8\,c_3/L_1^{2} \,\log n)$, with $c_3$  as in Lemma \ref{lem:Dvoretzky}.
\end{lemma}

\begin{proof}
	Let us assume that the events  $\Omega$  and  $\Omega^{\prime}$  from Lemma  \ref{lem:concentration_to_g} and Lemma \ref{lem:dist} hold.  Let also  $i,i^{\prime} \in [n \,-\, m]$  be such that 
	$\vert i\,-\, i^{\prime}\vert \,\,   >\,\,12(c_2\,n\,\log n)^{1/2}/L_1$.
	Then,
	% if $\xi_i \in A_l$ and $\xi_{i^{\prime} } \in A_{l^{\prime}}$  with  $l \neq l^{\prime}$,  $l,l ^{\prime}  \in [K]$, then
	%by Assumption  \ref{as5}  we have that 
	%\[
	%    c  \,\,<\,\,   \vert  g(\xi_i)  \,-\, g(\xi_{i^{\prime}}) \vert.   
	% \]
	% If $\xi_i,\xi_{i^{\prime}} \in A_l$  for some $l \in A_l$.  This implies that
	\[
	8\,\left(  \frac{c_2\,\log n}{n} \right)^{ \frac{1}{2} } \,\,\leq \,\, L_1\,\left\vert  \xi_{\sigma(i)}\,-\, \xi_{\sigma(i^{\prime})} \right\vert  \,\,\leq \,\,  \left\vert  g(\xi_{\sigma(i)})  \,-\,  g(\xi_{\sigma(i^{\prime})}) \right\vert.
	\]
	
	Therefore,  for large enough $n$, if   $\vert i\,-\, i^{\prime}\vert \,\,   >\,\,12(c_2\,n\,\log n)^{1/2}/L_1$   then
	\begin{equation}
	\label{eqn:iff}
	\underset{j\in [m]}{\sum} \,A_{\sigma(i),j}   \,<\,\underset{j\in [m]}{\sum} \,A_{\sigma(i^{\prime}),j}     \,\,\,\,\,\,\,\, \text{iff}  \,\,\,\,\,\,\,\,  g(\xi_{  \sigma(i) })  \,<\, g(\xi_{\sigma(i{\prime})}).
	\end{equation}
	%	The latter is equivalent to:  if   $\vert \sigma^{-1}(\tau(i))\,-\,\sigma^{-1}(\tau(i^{\prime}))\vert \,\,   >\,\,12(n\,\log n)^{1/2}/L_1$   then 
	%		\begin{equation}
	%			\label{eqn:iff}
	%					\underset{j\in [m]}{\sum} \,A_{\tau(i),j}   \,<\,\underset{j\in [m]}{\sum} \,A_{\tau(i^{\prime}),j}     \,\,\,\,\,\,\,\, \text{iff}  \,\,\,\,\,\,\,\,  g(\xi_{  \tau(i) })  \,<\, g(\xi_{\tau(i{\prime})}).
	%		\end{equation} 
	Next let us fix $i \in [n - m]$ and define 
	\[
	\Lambda_1\,=\, \left \{ i^{\prime}   \in [n - m]\,:\, \underset{j\in [m]}{\sum} \,A_{\sigma(i),j}   \,<\,\underset{j\in [m]}{\sum} \,A_{\sigma(i^{\prime}),j}  \right \},\,\,\,\,\,\,\,\,\,\,  \Lambda_2\,=\, \left \{ i^{\prime}   \in [n -  m]\,:\,  g(\xi_{  \sigma(i) })  \,<\, g(\xi_{\sigma(i{\prime})})  \right \}.
	\]
	Then if  (\ref{eqn:iff}) holds then  %$\vert  \Lambda_1  \backslash \Lambda_2 \vert \,\leq \, 24(c_2\,n\,\log n)^{1/2}/L_1 $. And so
	\[
	\left\vert  \hat{\tau}^{-1}(\sigma(i)) \,-\, \tau^{-1}(\sigma(i))    \right\vert \,= \, \left\vert   \vert  \vert \Lambda_1 \vert   \,-\, \vert \Lambda_2 \vert   \right\vert \,\leq \,  \max\{ \vert  \Lambda_1  \backslash \Lambda_2 \vert  , \vert  \Lambda_2  \backslash \Lambda_1 \vert  \} \,\leq \, 24(c_2\,n\,\log n)^{1/2}/L_1. 
	\]
	%	 Hence, for any $j  \in [n-m]$  setting $i\,=\,  \sigma^{-1}(\hat{\tau}(j)) $  and $j^{\prime }  \,=\, \tau^{-1}(\sigma(i)) $, we obtain that
	%	 \[
	%	    \vert  \xi_{\hat{\tau}(j)}  \,-\,\xi_{ \tau(j) }\vert 
	%	 \]
	%$j \,=\, \hat{\tau}^{-1}(\sigma(i))   $  and  $j^{\prime}  \,=\, \tau^{-1}(\sigma(i))  $, we have 
	
\end{proof}

\begin{theorem}
	\label{lem:tv_control}%\ref{as5}  \ref{sec:class_of_estimators}
	Let us suppose that Assumption 3  holds.  Let $\hat{\tau}$ be constructed  as in Section 2.2  by setting $d \,:=\, \hat{d}_1$. Then,
	%using the metric $\hat{d}_I$  as explained in 
	\[
	\displaystyle    \frac{1}{n^2}   \sum_{i=  1}^{n-m  - 1}  \sum_{j=m+1}^{n} \vert  \theta_{ \hat{\tau}(i),j}^* \,-\, \theta_{\hat{\tau}(i+1),j}^*\vert    \,\,=\,\, O_{\mathbb{P}}\left(   \sqrt{\frac{\log n}{n}}   \right).
	\]
\end{theorem}

\begin{proof}
	
	Let  $\delta  \,\,=\,\,   24 \,( c_2 \,\log n)^{1/2}/(L_1\,n^{1/2} )$,   $\mathcal{S}^{\prime} \,=\,  \{b_1,\ldots,b_r\}$  and  $\Lambda \,=\, \{  i\,:\,  \xi_{\tau(i)} \notin B_{\delta}(\mathcal{S}^{\prime})   \}$,  where  $B_{\delta}(\mathcal{S}^{\prime})  = \{  x \,:\,  \vert x - x^{\prime}\vert < \delta,  \,\,\,\text{for some } \,\,x^{\prime} \in     \mathcal{S}^{\prime}\}$ . We also write
	\[
	T_i\,=\,  \vert  \left\{  j  \in [n] \backslash [m] \,:\,  \vert  \xi_j   \,-\, \xi_{ \tau(i) }    \vert   \,\leq
	\, \delta     \right\}  .
	\]
	Then by  Proposition 27 from \cite{von2010hitting} we have that 
	\[
	\mathbb{P} \left(    \underset{i \in [n-m]}{\min} \,T_i  \,\geq  \, (n-m)\,\delta \right) \,\geq \, 1  \,-\,   (n-m)\exp\left( -\frac{n\,\delta}{6} \right),
	\]
	And so we assume that 
	$$ \underset{i \in [n-m]}{\min} \,T_i  \,\geq  \, (n-m)\,\delta$$
	holds. Let us also assume that the  events $\Omega_2$  from Lemma \ref{lem:comparison_permutations}  and $\Omega^{\prime}$ from Lemma \ref{lem:Dvoretzky}  hold. By Assumption 3,  we have that,  if  $i  \in \Lambda$  and $i^{\prime } \,=\, \tau^{-1}( \hat{\tau}(i) )  $,  then
	\[
	\begin{array}{lll}
	\vert  \xi_{ \tau(i) } \,-\, \xi_{ \hat{\tau}(i) }  \vert & =&\vert  \xi_{ \tau(i) } \,-\, \xi_{ \tau(i^{\prime}) }  \vert \\
	&\leq &\left\vert  \xi_{ \tau(i) } \,-\,    \frac{ \sigma^{-1}( \tau(i) ) }{n-m} \right\vert  \,+\, \left\vert  \frac{ \sigma^{-1}( \tau(i) ) }{n-m}\,-\,    \frac{ \sigma^{-1}( \tau(i^{\prime}) ) }{n-m} \right\vert  \,+\, \left\vert  \xi_{ \tau(i^{\prime}) } \,-\,    \frac{ \sigma^{-1}( \tau(i^{\prime}) ) }{n-m} \right\vert \\
	& \leq& \frac{4\, c_2^{ \frac{1}{2}  }}{L_1}\left( \frac{\log n}{n} \right)^{ \frac{1}{2} }   \,+\,  \left\vert  \frac{ i  }{n-m}\,-\,    \frac{ i^{\prime} }{n-m} \right\vert   \\
	&\leq &   \frac{28\, c_2^{ \frac{1}{2}  }}{L_1}\left( \frac{\log n}{n} \right)^{ \frac{1}{2} },  
	\end{array}
	\]
	where in the second inequality   we have used the fact that   $\xi_{\tau(i)}$  and  $\xi_{\tau(i^{\prime})}$   belong to the same connected component of $[0,1] \backslash B_{\delta}(\mathcal{S}^{\prime})$,  and the function  $g$ is  strictly piecewise monotonic. 
	
	The argument above implies that
	\[
	\underset{i \in  \Lambda }{ \max }  \,\left\vert   \xi_{\hat{\tau}(i)} \,\,-\,\,\frac{ \sigma^{-1}(\tau(i))  }{n} \right\vert  \,\,\leq\,\,\frac{30}{L_1}\left( \frac{c_2\,\log n}{n} \right)^{ \frac{1}{2} }.
	\] 
	
	We now let  $\delta_2  \,\,=\,\,   150 \,(\log n)^{1/2}/(L_1\,n^{1/2} )$. 	We also  let  $s  \,=\,  \ceil{ \delta_2\,n }$, and clearly  $s \, \asymp  \, \sqrt{n\,\log n} $. Moreover, for  $i \in \{0,1,\ldots,s-1\}$  we define $l_i \,=\,\floor{   (1 - i/n)\,/(s/n)  } $  and
	\[
	\mathcal{B}_0^i\,=\,   \left[0,\frac{i}{n}\right),\,\,\,  \mathcal{B}_1^i\,=\,   \left[\frac{i}{n},\frac{i}{n}+\frac{s}{n}\right),\,\,\,\ldots,\,\,\, \mathcal{B}_{l_i}^i\,=\,   \left[\frac{i}{n} +\frac{s\,(l_i-1) }{n} ,\frac{i}{n}+\frac{s\,l_i}{n}\right),  \,\,\, \mathcal{B}_{l_i+1}^i\,=\,   \left[\frac{i}{n} +\frac{s\,l_i }{n} ,1\right].
	\]
	Then  if  $y_k^i,z_k^i \in \mathcal{B}_k^i$   with  $y_k^i  \,<\,z_k^i  $  and  $k \in \{1,\ldots,l_i \}$,   Assumption 3 implies that  for any $j \in [n] \backslash [m]$ 
	\begin{equation}
	\label{eqn:BV}
	\displaystyle \sum_{k=1}^{l_i} \left[  \left\vert  f_0\left( \frac{i  +  (k-1)s }{n} ,\xi_j\right) \,-\, f_0\left(  y_k^i,\xi_j\right)\right\vert  \,\,+\,\, \left\vert  f_0\left(  y_k^i,\xi_j\right) \,-\, f_0\left(  z_k^i,\xi_j\right)\right\vert   \,\,+\,\, \left\vert   f_0\left(  z_k^i,\xi_j\right)\,-\,f_0\left( \frac{i+ks}{n},\xi_j \right) \right\vert   \right] \,\,\,\leq\,\,\,C.
	\end{equation}
	
	On the other hand,  with probability approaching one,
	%\[
	\begin{equation}
	\label{eqn:aux1}
	\left\vert  \xi_{ \hat{\tau}(i^{\prime} )} \,-\,  \frac{ \sigma^{-1}(\tau(i^{\prime}))     }{n}   \right\vert  \,\,\leq \,\, \frac{\delta}{4}\,\,\leq \,\,\frac{s}{4\,n},\,\,\,\,\,\,\,\,\,\, \left\vert  \xi_{ \hat{\tau}(i^{\prime}+1)} \,-\, \frac{ \sigma^{-1}(\tau(i^{\prime}))     }{n}  \right\vert  \,\,\leq \,\, \frac{\delta}{4}\,\,\leq \,\,\frac{s}{4\,n},\,\,\,\,\,\,\,\,\,\,\,\,\forall i^{\prime} \in \Lambda.
	\end{equation}
	Next let $I \,=\,  \{  i^{\prime} \in [n-m]  \,\,:\,\,  \xi_{ \hat{\tau}(i^{\prime}) },\xi_{ \hat{\tau}(i^{\prime}+1) } \in  (\frac{1}{n} + \frac{s}{2\,n},  1 \,-\,\frac{2s}{n}  )    \}$. 
	Let us assume that (\ref{eqn:aux1}) holds  and let $i^{\prime} \in  I$. Then we set   $i^{\prime \prime} \,=\,   \floor{  \sigma^{-1}(\tau(i^{\prime}))  \,-\,  \frac{s}{4}  }$ and observe that
	%  and defining $i^{\prime \prime} \,=\,   \floor{ i^{\prime} \,-\,  \frac{s}{4}  }$ implies  
	\[
	\xi_{ \hat{\tau}(i^{\prime})  }, \,\xi_{ \hat{\tau}(i^{\prime} +1)  }   \,\in \, \left( \frac{ i^{\prime \prime}    }{n} \,,\,  \frac{ i^{\prime\prime}    }{n}   \,+\, \frac{s}{n}  \right)  \,=\,  \mathcal{B}^{r_{i^{\prime} } }_{k_{i^{\prime} } },
	\]
	where  $\floor{  \sigma^{-1}(\tau(i^{\prime})) \,-\,  \frac{s}{4}  }\,=\,  r_{i^{\prime}}  \,+\,   k_{i^{\prime}}\,s $  with  $r_{i^{\prime}} \in \{0,1,\ldots,s-1\}$  and  $k_{i^{\prime}} \in [l_{r_{i^{\prime}}}]$. Therefore using (\ref{eqn:BV}), with probability approaching one we have that
	\[
	\begin{array}{l}
	\displaystyle  \sum_{j \in  [n]\backslash [m] }  \sum_{i=1}^{n-m-1} \,\vert f_0(\xi_{ \hat{\tau}(i) },\xi_j)  \,-\, f_0(\xi_{ \hat{\tau}(i+1) },\xi_j)  \vert\\
	\leq \,\,\,  2\,n\,\| f_0\|_{\infty}\,\left(\vert [n] \backslash  I\vert  \,+\,\vert [n] \backslash  \Lambda\vert \right)\,\,\,+\,\,\,   \displaystyle \sum_{j \in  [n]\backslash [m] } \sum_{i \in I \cap \Lambda }   \,\vert f_0(\xi_{ \hat{\tau}(i) },\xi_j)  \,-\, f_0(\xi_{ \hat{\tau}(i+1) },\xi_j)  \vert\\
	\leq \,\,\,  \displaystyle   \sum_{j \in  [n]\backslash [m] } \,\,\underset{ \,  y_k^i\, < \, z_k^i  \in \mathcal{B}_k^i,\,\,\, i \in \{0,1,,\ldots,s-1\},\,\,  k \in [l_i]  }{\sup}\,\,\, \sum_{i=0}^{s-1}\, \,\sum_{k=1}^{l_i} \Bigg[   \,\left\vert  f_0\left(  y_k^i,\xi_j\right) \,-\, f_0\left(  z_k^i,\xi_j\right)\right\vert      \Bigg]   \,\,+\,\, \\
	\,\,\,\,\,\, \,\,\, \,\,\,  2\,n\,\| f_0\|_{\infty}\,\left(\vert[n]\backslash I \vert  \,+\,\vert [n] \backslash \Lambda \vert \right)\\
	\leq  \,\,\,2\,n\,\| f_0\|_{\infty}\,\left(\vert [n] \backslash I\vert  \,+\,\vert [n] \backslash \Lambda \vert \right)  \,\,+\,\,  n\,s\,C.
	\end{array}
	\]
	Hence, by the binomial concentration inequality we have that 
	\[
	\vert [n] \backslash I\vert  \,\,+\,\,   \vert [n] \backslash \Lambda \vert   \,\,=\,\, O_{ \mathbb{P} }\left( \sqrt{ n \,\log n}  \right).
	\]
\end{proof}

\paragraph{Proof  of Theorem  2: } This follows from the previous theorem  combined with Lemma \ref{lem:tv:rates}.

\subsection{Proof of Theorem 1 }%\ref{thm:bound_on_tv} }

As in the proof Lemma  2  from \cite{zhang2015estimating}  we obtain the following result:

\begin{lemma}
	\label{lem:uniform_bound}
	Assume that $m  \asymp  n$, then there exists  positive constants $C_1$  and  $C_2$  such that with probability  at least   $1  \,-\,  n^{-C_2}  $  the event  
	\[
	\frac{1}{n}\,\,\, \underset{ i,j \in [n],  \,\,i \neq j      }{\max } \, \vert  \langle  A_{i,[m]}, A_{j,[m]} \rangle  \,-\,  \langle  \theta^*_{i,[m]}, \theta^*_{j,[m]} \rangle  \vert \,\,\leq \,\,  C_1\left(   \frac{\log n}{n} \right)^{ \frac{1}{2} } 
	\]
	holds.
	%	Let    $d$  be the distance given by
	%	\[
	%	     d(i,i^{\prime}) \,\,=\,\,   \| A_{i,\cdot}\,-\, A_{i^{\prime},\cdot}     \|_1  \,\,=\,\,   \displaystyle \sum_{j=1}^{n} \,\vert  A_{i,j}  \,-\, A_{i^{\prime},j} \vert.
	%	\]
\end{lemma}

As a  consequence, we obtain our next result.

\begin{lemma}
	\label{cor:inner_product_dist_bound}
	With  $\hat{d}_I$  defined as in  (7) in the main paper,
	there  exist   positive  constants $C_1$  and  $C_2$  such that%With the same constants  from  Lemma  \ref{lem:uniform_bound}, 
	% we have that 
	\[
	\vert  \hat{d}_I(i,i^{\prime})    \,-\,  d_I(i,i^{\prime}) \vert  \,\,\leq  \,\,  \sqrt{ 2\,C_1} \left(   \frac{\log n}{n} \right)^{ \frac{1}{4} },\,\,\,\,\forall   i,i^{\prime} \in [n] \backslash [m], 
	\]
	with probability at least  $1  \,-\,  n^{-C_2}  $  where
	\[
	d_I(i,i^{\prime}) \,\,=\,\,  \underset{  k   \in [m] }{\max }\sqrt{   \frac{1}{m} \vert  \langle  \theta^*_{i, [m]},  \theta^*_{k, [m] } \rangle    \,-\,    \langle  \theta^*_{i^{\prime}, [m]},  \theta^*_{k, [m] } \rangle   \vert      }, \,\,\,\,\,\forall  i, i^{\prime}  \in [n] \backslash [m].
	\]
\end{lemma}

\begin{lemma}
	\label{lem:bar_dist}
	Let  $\bar{d}_I$  be a distance between  $i,i^{\prime}  \in [n] \backslash [m]$  defined as
	\[
	\bar{d}_I(i,i^{\prime})\,\,=\,\,   \underset{   k \in [m]  }{\max}   \,\sqrt{ \left\vert     \displaystyle \int_0^1    \left[ f_0(\xi_i,t) \,-\,   f_0(\xi_{i^{\prime}},t)    \right] f_0(\xi_k,t)  \,dt  \right\vert   }.     
	\]
	Then there exist positive constants  $c_1,c_2 >0$  such that for large enough  $n$ with probability at least  $1 \,-\,  n^{-c_2} $  we have that  	
	\[
	\left\vert \bar{d}_I(i,i^{\prime})  \,-\, \hat{d}_I(i,i^{\prime})  \right\vert    \,\,\leq \,\,    c_1 \left(   \frac{\log n}{n} \right)^{1/4}  \,\,\,\,\,\,\,\, \forall  i,i^{\prime}  \in  [n] \backslash [m]. 
	\]
\end{lemma}

\begin{proof}
	We notice that
	\[
	\def\arraystretch{2.2}
	\begin{array}{lll}
	\left\vert \bar{d}_I(i,i^{\prime})  \,-\, \hat{d}_I(i,i^{\prime})  \right\vert^2   &\leq  & \left\vert \bar{d}_I(i,i^{\prime})^2  \,-\, \hat{d}_I(i,i^{\prime})^2  \right\vert \\
	& \leq &  \Bigg\vert  \underset{   k \in [m]  }{\max}   \, \left\vert     \displaystyle \int_0^1    \left[ f_0(\xi_i,t) \,-\,   f_0(\xi_{i^{\prime}},t)    \right] f_0(\xi_k,t)  \,dt  \right\vert     \,-\,\\
	& &\,\displaystyle  \underset{  k   \in [m] }{\max } \,  \frac{1}{m} \vert  \langle  \theta^*_{i, [m]},  \theta^*_{k, [m] } \rangle    \,-\,    \langle  \theta^*_{i^{\prime}, [m]},  \theta^*_{k, [m] } \rangle   \vert      \Bigg\vert \\
	& \leq&  \underset{   k \in [m]  }{\max}  \Bigg\vert  \displaystyle \int_0^1    \left[ f_0(\xi_i,t) \,-\,   f_0(\xi_{i^{\prime}},t)    \right] f_0(\xi_k,t)  \,dt \,-\,\\
	& &\,\displaystyle \frac{1}{m} \,  \,    \langle  \theta^*_{i, [m]},  \theta^*_{k, [m] } \rangle    \,-\,    \langle  \theta^*_{i^{\prime}, [m]},  \theta^*_{k, [m] } \rangle\Bigg\vert.
	\end{array}
	\]
	The claim  follows  by   using  Theorem \ref{thm:Hoeffding}  to bound  
	\[
	\Bigg\vert  \displaystyle \int_0^1    \left[ f_0(\xi_i,t) \,-\,   f_0(\xi_{i^{\prime}},t)    \right] f_0(\xi_k,t)  \,dt \,-\,\\
	\,\displaystyle \frac{1}{m} \,  \,    \langle  \theta^*_{i, [m]},  \theta^*_{k, [m] } \rangle    \,-\,    \langle  \theta^*_{i^{\prime}, [m]},  \theta^*_{k, [m] } \rangle\Bigg\vert
	\]
	for any $i,i^{\prime} \in [n] \backslash [m]$,  $k \in [m]$. The latter  can  be done by conditioning on $\xi_i$, $\xi_{i^{\prime}}$  and $\xi_{k}$, allowing us to obtain a uniform  bound.
	
\end{proof}

\begin{theorem}
	\label{thm:bound_on_tv}
	Let us suppose that Assumptions 1-2 hold. Then,  we  extend $\theta^*$  to be in $\mathbb{R}^{ n \times n}$, by setting $\theta*_{j,i}  \,=\,   \theta^*_{i,j}$  for all
	$i < j$  with  $i,j \in [n] \backslash [m]$. Moreover,  we set  $\theta_{i,i}^*   \,= \,  0$   for all  $i \in [n] \backslash [m]$. Let $\hat{\tau}$ be constructed  as in Section 2.2 by setting $d \,:=\, \hat{d}_I$. Then,
	%using the metric $\hat{d}_I$  as explained in 
	\[
	\displaystyle    \frac{1}{n^2}   \sum_{i=  1}^{n-m  - 1}  \sum_{j=m+1}^{n} \vert  \theta_{ \hat{\tau}(i),j}^* \,-\, \theta_{\hat{\tau}(i+1),j}^*\vert    \,\,=\,\, O_{\mathbb{P}}\left(     \frac{  \log^{\frac{(1+\alpha)}{2}   }n  }{n^{  \frac{\alpha}{2} }}  \right).
	\]
\end{theorem}

\begin{proof}
	We first notice that by  Theorem \ref{thm:Hoeffding}  we have that  with probability approaching one
	\begin{equation}
	\label{eqn:TV1}
	\underset{i, i^{\prime} \in  [n] \backslash [m]   }{\max}  \left\vert \frac{1}{n-m}  \displaystyle \sum_{j=m+1 }^{n-1} \vert  \theta_{i,j}^* \,-\, \theta_{i^{\prime},j}^*\vert \,\,-\,\,   \displaystyle  \int_{0}^1   \vert  f_0(\xi_{i^{\prime}},t)  \,-\,  f_0(\xi_{i},t)  \vert dt      \right \vert    \,\,\leq \,\,C_1\left( \frac{\log n}{n}  \right)^{1/2},      
	\end{equation}
	for some positive constant $C_1 >0$ (we can attain this by first conditioning on $\xi$).
	
	On the other hand, with high probability,  for any $i \in [n] \backslash [m]$  by Assumption 1  there exists $k_i \in [m]$  such that
	\begin{equation}
	\label{eqn:aux}
	\left\vert  \left\vert   \displaystyle  \int_{0}^1   (  f_0(\xi_{i^{\prime} },t)  \,-\,  f_0(\xi_{i},t)) f_0(\xi_{i },t) dt      \right \vert  \,\,-\,\,\left\vert   \displaystyle  \int_{0}^1   (  f_0(\xi_{i^{\prime} },t)  \,-\,  f_0(\xi_{i},t)) f_0(\xi_{k_i },t) dt      \right \vert  \right\vert \,\,\leq \,\,  2\|f_0\|_{\infty}c_1  \left(  \frac{\log n}{n}\right)^{\alpha}
	\end{equation}
	
	Next, for any $i, i^{\prime} \in [n ]\backslash [m]$  we have that
	\begin{equation}
	\label{eqn:aux2}
	\begin{array}{lll}
	\displaystyle  \int_{0}^1   \vert  f_0(\xi_{i^{\prime} },t)  \,-\,  f_0(\xi_{i},t)  \vert dt     & \leq &  \left[ \displaystyle  \int_{0}^1   \vert  f_0(\xi_{i^{\prime}},t)  \,-\,  f_0(\xi_{i},t)  \vert^2 dt   \right]^{1/2} \\
	&  = &\left[ \displaystyle  \int_{0}^1   (  f_0(\xi_{i^{\prime} },t)  \,-\,  f_0(\xi_{i},t)) f_0(\xi_{i^{\prime} },t) dt  \,\,+\,\,\int_{0}^1   (  f_0(\xi_{i^{\prime} },t)  \,-\,  f_0(\xi_{i},t)) f_0(\xi_{i },t) dt   \right]^{1/2}\\
	&  \leq  &   \left \vert \displaystyle \int_{0}^1   (  f_0(\xi_{i^{\prime} },t)  \,-\,  f_0(\xi_{i},t)) f_0(\xi_{i^{\prime} },t) dt  \right\vert^{1/2}   \,+\,\left \vert \displaystyle \int_{0}^1   (  f_0(\xi_{i^{\prime} },t)  \,-\,  f_0(\xi_{i},t)) f_0(\xi_{i },t) dt  \right\vert^{1/2}\\
	& \leq &  2\left[   \bar{d}_I(i,i^{\prime})^2 \,+\,   2\|f_0\|_{\infty}c_1  \left(  \frac{\log n}{n}\right)^{\alpha}  \right]^{1/2}\\
	& \leq &2\bar{d}_I(i,i^{\prime})   \,\,+  \,\, 2\sqrt{ 2\|f_0\|_{\infty}c_1}  \left(  \frac{\log n}{n}\right)^{\alpha/2}.
	\end{array}
	\end{equation}
	
	Therefore,   combining  (\ref{eqn:TV1}), (\ref{eqn:aux})  and (\ref{eqn:aux2})    we obtain  that with probability approaching $1$
	\begin{equation}
	\label{aux3}
	\begin{array}{lll}
	\displaystyle    \frac{1}{n}   \sum_{i=  1}^{n-m  - 1}  \sum_{j=m+1}^{n} \vert  \theta_{ \hat{\tau}(i),j}^* \,-\, \theta_{\hat{\tau}(i+1),j}^*\vert   & \leq &   \displaystyle    \sum_{i=1}^{n-m  - 1} \left[  \int_{0}^1   \vert  f_0(\xi_{\hat{\tau}(i) },t)  \,-\,  f_0(\xi_{\hat{\tau}(i+1)},t)  \vert dt  \,\,+\,\,C_1\left( \frac{\log n}{n}  \right)^{1/2} \right]\\
	& \leq &\displaystyle    \sum_{i=1}^{n-m  - 1}  \left[ 2\bar{d}_I(\hat{\tau}(i),\hat{\tau}(i+1))   \,\,+  \,\, 2\sqrt{ 2\|f_0\|_{\infty}c_1}  \left(  \frac{\log n}{n}\right)^{\alpha/2}  \right]  \\
	& &\,\,+\,\,
	C_1\left( n\,\log n  \right)^{1/2}\\
	& \leq  & \displaystyle  2   \sum_{i=1}^{n-m-1}  \left[  \hat{d}_I(\hat{\tau}(i),\hat{\tau}(i+1)) \,+\,    c_1 \left(   \frac{\log n}{n} \right)^{1/4} \right] \\%\,\,+\,\,
	%		   	    \\
	& & \,\,\,+\,\, C_1\left( n\,\log n  \right)^{1/2}   \,\,+\,\,  2\sqrt{ 2\|f_0\|_{\infty}c_1} \,n^{1 - \alpha/2}  \log^{\alpha/2} n   \\
	% \int_{0}^1   \vert  f_0(\xi_{i^{\prime}},t)  \,-\,  f_0(\xi_{i},t)  \vert dt
	& \leq & 2\left( 1 + \frac{\log_2 n   }{2}  \right)   \displaystyle \sum_{i=1}^{n-m-1} \hat{d}_I(\tau^*(i),\tau^*(i+1))  \,\,+\,\, C_1\left( n\,\log n  \right)^{1/2}\\
	& &  \,\,\,+\,\,  2c_1\,n^{3/4} \log^{1/4} n   \,\,+\,\,2\sqrt{ 2\|f_0\|_{\infty}c_1} \,n^{1 - \alpha/2}  \log^{\alpha/2} n 
	\end{array}
	\end{equation}
	where the third inequality follows from Lemma \ref{lem:bar_dist}  and the last one from Theorem 1  in \cite{rosenkrantz1977intussusception}. Next, we proceed to bound the last term in the previous inequality. To that end, we notice that  by Lemma \ref{lem:bar_dist},  with probability approaching one,  for all $i \in [n-m-1]$ it holds that
	\[
	\begin{array}{lll}
	\hat{d}_I(\tau^*(i),\tau^*(i+1))   &\leq &  \underset{   k \in [m]  }{\max}   \,\sqrt{ \left\vert     \displaystyle \int_0^1    \left[ f_0(\xi_{ \tau^*(i) },t) \,-\,   f_0(\xi_{\tau^*(i+1) },t)    \right] f_0(\xi_k,t)  \,dt  \right\vert   }   \,\,+\,\, \kappa_1 \left(   \frac{\log n}{n} \right)^{1/4}\\
	& \leq &  \sqrt{   \displaystyle \int_0^1    \vert f_0(\xi_{ \tau^*(i) },t) \,-\,   f_0(\xi_{\tau^*(i+1) },t)  \vert   \,dt }\,\,+\,\, \kappa_1 \left(   \frac{\log n}{n} \right)^{1/4},\\
	%& \leq &
	\end{array}
	\]
	for some positive constant  $\kappa_1$. Therefore, combining with (\ref{aux3})
	\[
	\begin{array}{lll}
	\displaystyle    \frac{1}{n}   \sum_{i=  1}^{n-m  - 1}  \sum_{j=m+1}^{n} \vert  \theta_{ \hat{\tau}(i),j}^* \,-\, \theta_{\hat{\tau}(i+1),j}^*\vert   & \leq &   2\left( 1 + \frac{\log_2 n   }{2}  \right)   \displaystyle \sum_{i=1}^{n-m-1} \sqrt{   \displaystyle \int_0^1    \vert f_0(\xi_{ \tau^*(i) },t) \,-\,   f_0(\xi_{\tau^*(i+1) },t)  \vert   \,dt }\\ 
	& &  \,\,+\,\,\left( 2c_1 \kappa_1  \,+\, C_1 \right) n^{3/4} \log^{1/4} n   \,\,+\,\, 2\sqrt{ 2\|f_0\|_{\infty}c_1} \,n^{1 - \alpha/2}  \log^{\alpha/2} n  \\
	& \leq  & 2\left( 1 + \frac{\log_2 n   }{2}  \right)   \sqrt{n}\,\left(    \displaystyle \sum_{i=1}^{n-m-1}  \displaystyle \int_0^1    \vert f_0(\xi_{ \tau^*(i) },t) \,-\,   f_0(\xi_{\tau^*(i+1) },t)  \vert   \,dt \right)^{1/2}\\ 
	& &  \,\,+\,\,\left( 2c_1 \kappa_1  \,+\, C_1 \right) n^{3/4} \log^{1/4} n   \,\,+\,\, 2\sqrt{ 2\|f_0\|_{\infty}c_1} \,n^{1 - \alpha/2}  \log^{\alpha/2} n,  \\	 	   	 
	%	  & &  \,\,+\,\,\left( 2c_1 \,+\,  2\sqrt{ 2\|f_0\|_{\infty}c_1}  \,+\, \kappa_1  \,+\, C_1 \right) n^{3/4} \log^{1/4} n \\
	\end{array}
	\] 
	with probability approaching $1$. The conclusion follows from Assumption 2.
	% with probability approaching $1$.
	%	rosenkrantz1977intussusception theorem 1
\end{proof}

\paragraph{Proof  of Theorem  1: } This follows combining Theorem \ref{thm:bound_on_tv}  with Lemma \ref{lem:tv:rates}.

	\bibliographystyle{abbrvnat}
	\bibliography{graphfused}

\begin{thebibliography}{35}
\providecommand{\natexlab}[1]{#1}
\providecommand{\url}[1]{\texttt{#1}}
\expandafter\ifx\csname urlstyle\endcsname\relax
  \providecommand{\doi}[1]{doi: #1}\else
  \providecommand{\doi}{doi: \begingroup \urlstyle{rm}\Url}\fi

\bibitem[Adamic and Glance(2005)]{adamic2005political}
Lada~A Adamic and Natalie Glance.
\newblock The political blogosphere and the 2004 us election: divided they
  blog.
\newblock In \emph{Proceedings of the 3rd international workshop on Link
  discovery}, pages 36--43. ACM, 2005.

\bibitem[Airoldi et~al.(2013)Airoldi, Costa, and Chan]{airoldi2013stochastic}
Edo~M Airoldi, Thiago~B Costa, and Stanley~H Chan.
\newblock Stochastic blockmodel approximation of a graphon: Theory and
  consistent estimation.
\newblock In \emph{Advances in Neural Information Processing Systems}, pages
  692--700, 2013.

\bibitem[Barbero and Sra(2014)]{barbero2014modular}
{\'A}lvaro Barbero and Suvrit Sra.
\newblock Modular proximal optimization for multidimensional total-variation
  regularization.
\newblock \emph{arXiv preprint arXiv:1411.0589}, 2014.

\bibitem[Bellmore and Nemhauser(1968)]{bellmore1968traveling}
Mandell Bellmore and George~L Nemhauser.
\newblock The traveling salesman problem: a survey.
\newblock \emph{Operations Research}, 16\penalty0 (3):\penalty0 538--558, 1968.

\bibitem[Bickel and Chen(2009)]{bickel2009nonparametric}
Peter~J Bickel and Aiyou Chen.
\newblock A nonparametric view of network models and newman--girvan and other
  modularities.
\newblock \emph{Proceedings of the National Academy of Sciences}, 106\penalty0
  (50):\penalty0 21068--21073, 2009.

\bibitem[Boucheron et~al.(2013)Boucheron, Lugosi, and
  Massart]{boucheron2013concentration}
St{\'e}phane Boucheron, G{\'a}bor Lugosi, and Pascal Massart.
\newblock \emph{Concentration inequalities: A nonasymptotic theory of
  independence}.
\newblock Oxford university press, 2013.

\bibitem[Butland et~al.(2005)Butland, Peregr{\'\i}n-Alvarez, Li, Yang, Yang,
  Canadien, Starostine, Richards, Beattie, Krogan,
  et~al.]{butland2005interaction}
Gareth Butland, Jos{\'e}~Manuel Peregr{\'\i}n-Alvarez, Joyce Li, Wehong Yang,
  Xiaochun Yang, Veronica Canadien, Andrei Starostine, Dawn Richards, Bryan
  Beattie, Nevan Krogan, et~al.
\newblock Interaction network containing conserved and essential protein
  complexes in escherichia coli.
\newblock \emph{Nature}, 433\penalty0 (7025):\penalty0 531, 2005.

\bibitem[Chan and Airoldi(2014)]{chan2014consistent}
Stanley Chan and Edoardo Airoldi.
\newblock A consistent histogram estimator for exchangeable graph models.
\newblock In \emph{International Conference on Machine Learning}, pages
  208--216, 2014.

\bibitem[Chatterjee et~al.(2015)]{chatterjee2015matrix}
Sourav Chatterjee et~al.
\newblock Matrix estimation by universal singular value thresholding.
\newblock \emph{The Annals of Statistics}, 43\penalty0 (1):\penalty0 177--214,
  2015.

\bibitem[Clarkson and Adams(1933)]{clarkson1933definitions}
James~A Clarkson and C~Raymond Adams.
\newblock On definitions of bounded variation for functions of two variables.
\newblock \emph{Transactions of the American Mathematical Society}, 35\penalty0
  (4):\penalty0 824--854, 1933.

\bibitem[Gao and Ma(2018)]{gao2018minimax}
Chao Gao and Zongming Ma.
\newblock Minimax rates in network analysis: Graphon estimation, community
  detection and hypothesis testing.
\newblock \emph{arXiv preprint arXiv:1811.06055}, 2018.

\bibitem[Gao et~al.(2015)Gao, Lu, Zhou, et~al.]{gao2015rate}
Chao Gao, Yu~Lu, Harrison~H Zhou, et~al.
\newblock Rate-optimal graphon estimation.
\newblock \emph{The Annals of Statistics}, 43\penalty0 (6):\penalty0
  2624--2652, 2015.

\bibitem[Gao et~al.(2016)Gao, Lu, Ma, and Zhou]{gao2016optimal}
Chao Gao, Yu~Lu, Zongming Ma, and Harrison~H Zhou.
\newblock Optimal estimation and completion of matrices with biclustering
  structures.
\newblock \emph{Journal of Machine Learning Research}, 17\penalty0
  (161):\penalty0 1--29, 2016.

\bibitem[Gu{\'e}don and Vershynin(2016)]{guedon2016community}
Olivier Gu{\'e}don and Roman Vershynin.
\newblock Community detection in sparse networks via grothendieck’s
  inequality.
\newblock \emph{Probability Theory and Related Fields}, 165\penalty0
  (3-4):\penalty0 1025--1049, 2016.

\bibitem[Hoeffding(1963)]{hoeffding1963probability}
Wassily Hoeffding.
\newblock Probability inequalities for sums of bounded random variables.
\newblock \emph{Journal of the American statistical association}, 58\penalty0
  (301):\penalty0 13--30, 1963.

\bibitem[Hutter and Rigollet(2016)]{hutter2016optimal}
Jan-Christian Hutter and Philippe Rigollet.
\newblock Optimal rates for total variation denoising.
\newblock \emph{Annual Conference on Learning Theory}, 29:\penalty0 1115--1146,
  2016.

\bibitem[Klopp and Verzelen(2017)]{klopp2017optimal}
Olga Klopp and Nicolas Verzelen.
\newblock Optimal graphon estimation in cut distance.
\newblock \emph{arXiv preprint arXiv:1703.05101}, 2017.

\bibitem[Mammen and van~de Geer(1997)]{mammen1997locally}
Enno Mammen and Sara van~de Geer.
\newblock Locally apadtive regression splines.
\newblock \emph{Annals of Statistics}, 25\penalty0 (1):\penalty0 387--413,
  1997.

\bibitem[Padilla et~al.(2016)Padilla, Scott, Sharpnack, and
  Tibshirani]{padilla2016dfs}
Oscar Hernan~Madrid Padilla, James~G Scott, James Sharpnack, and Ryan~J
  Tibshirani.
\newblock The dfs fused lasso: Linear-time denoising over general graphs.
\newblock \emph{arXiv preprint arXiv:1608.03384}, 2016.

\bibitem[Rohe et~al.(2011)Rohe, Chatterjee, and Yu]{rohe2011spectral}
Karl Rohe, Sourav Chatterjee, and Bin Yu.
\newblock Spectral clustering and the high-dimensional stochastic blockmodel.
\newblock \emph{The Annals of Statistics}, pages 1878--1915, 2011.

\bibitem[Rosenkrantz et~al.(1977{\natexlab{a}})Rosenkrantz, Stearns, and
  Lewis]{rosenkrantz1977analysis}
Daniel~J Rosenkrantz, Richard~E Stearns, and Philip~M Lewis, II.
\newblock An analysis of several heuristics for the traveling salesman problem.
\newblock \emph{SIAM journal on computing}, 6\penalty0 (3):\penalty0 563--581,
  1977{\natexlab{a}}.

\bibitem[Rosenkrantz et~al.(1977{\natexlab{b}})Rosenkrantz, Cox, Silverman, and
  Martin]{rosenkrantz1977intussusception}
Jens~G Rosenkrantz, Joseph~A Cox, Frederic~N Silverman, and Lester~W Martin.
\newblock Intussusception in the 1970s: indications for operation.
\newblock \emph{Journal of pediatric surgery}, 12\penalty0 (3):\penalty0
  367--373, 1977{\natexlab{b}}.

\bibitem[Rudin et~al.(1992)Rudin, Osher, and Faterni]{rudin1992nonlinear}
Leonid Rudin, Stanley Osher, and Emad Faterni.
\newblock Nonlinear total variation based noise removal algorithms.
\newblock \emph{Physica {D}: Nonlinear Phenomena}, 60\penalty0 (1):\penalty0
  259--268, 1992.

\bibitem[Sadhanala et~al.(2016)Sadhanala, Wang, and
  Tibshirani]{sadhanala2016total}
Veeranjaneyulu Sadhanala, Yu-Xiang Wang, and Ryan~J. Tibshirani.
\newblock Total variation classes beyond 1d: {Minimax} rates, and the
  limitations of linear smoothers.
\newblock \emph{To appear, Neural Information Processing Systems}, 2016.

\bibitem[Song et~al.(2016)Song, Lee, Li, and Shah]{song2016blind}
Dogyoon Song, Christina~E Lee, Yihua Li, and Devavrat Shah.
\newblock Blind regression: Nonparametric regression for latent variable models
  via collaborative filtering.
\newblock In \emph{Advances in Neural Information Processing Systems}, pages
  2155--2163, 2016.

\bibitem[Tansey and Scott(2015)]{tansey2015fast}
Wesley Tansey and James Scott.
\newblock A fast and flexible algorithm for the graph-fused lasso.
\newblock \emph{arXiv preprint arXiv:1505.06475}, 2015.

\bibitem[Tibshirani et~al.(2005)Tibshirani, Saunders, Rosset, Zhu, and
  Knight]{tibshirani2005sparsity}
Robert Tibshirani, Michael Saunders, Saharon Rosset, Ji~Zhu, and Keith Knight.
\newblock Sparsity and smoothness via the fused lasso.
\newblock \emph{Journal of the Royal Statistical Society: Series B},
  67\penalty0 (1):\penalty0 91--108, 2005.

\bibitem[Tibshirani(2014)]{tibshirani2014adaptive}
Ryan~J. Tibshirani.
\newblock Adaptive piecewise polynomial estimation via trend filtering.
\newblock \emph{The Annals of Statistics}, 42\penalty0 (1):\penalty0 285--323,
  2014.

\bibitem[Tibshirani and Taylor(2011)]{tibshirani2011solution}
Ryan~J. Tibshirani and Jonathan Taylor.
\newblock The solution path of the generalized lasso.
\newblock \emph{Annals of Statistics}, 39\penalty0 (3):\penalty0 1335--1371,
  2011.

\bibitem[Von~Luxburg et~al.(2010)Von~Luxburg, Radl, and Hein]{von2010hitting}
Ulrike Von~Luxburg, Agnes Radl, and Matthias Hein.
\newblock Hitting and commute times in large graphs are often misleading.
\newblock \emph{arXiv preprint arXiv:1003.1266}, 2010.

\bibitem[Wang et~al.(2016)Wang, Sharpnack, Smola, and
  Tibshirani]{wang2016trend}
Yu-Xiang Wang, James Sharpnack, Alex Smola, and Ryan~J Tibshirani.
\newblock Trend filtering on graphs.
\newblock \emph{Journal of Machine Learning Research}, 17\penalty0
  (105):\penalty0 1--41, 2016.

\bibitem[Wolfe and Olhede(2013)]{wolfe2013nonparametric}
Patrick~J Wolfe and Sofia~C Olhede.
\newblock Nonparametric graphon estimation.
\newblock \emph{arXiv preprint arXiv:1309.5936}, 2013.

\bibitem[Xu(2017)]{xu2017rates}
Jiaming Xu.
\newblock Rates of convergence of spectral methods for graphon estimation.
\newblock \emph{arXiv preprint arXiv:1709.03183}, 2017.

\bibitem[Yan et~al.(2018)Yan, Sarkar, and Cheng]{yan2018provable}
Bowei Yan, Purnamrita Sarkar, and Xiuyuan Cheng.
\newblock Provable estimation of the number of blocks in block models.
\newblock In \emph{International Conference on Artificial Intelligence and
  Statistics}, pages 1185--1194, 2018.

\bibitem[Zhang et~al.(2015)Zhang, Levina, and Zhu]{zhang2015estimating}
Yuan Zhang, Elizaveta Levina, and Ji~Zhu.
\newblock Estimating network edge probabilities by neighborhood smoothing.
\newblock \emph{arXiv preprint arXiv:1509.08588}, 2015.

\end{thebibliography}

\end{document}